\documentclass[11pt]{article}
\usepackage{geometry}
\usepackage{setspace}
\usepackage[authoryear,round]{natbib}
\usepackage[utf8]{inputenc} 
\usepackage[hidelinks]{hyperref}      
\usepackage{url}   
\usepackage{enumitem}        
\usepackage{booktabs}       
\usepackage{amsfonts}      
\usepackage{nicefrac}      
\usepackage{microtype}     
\usepackage[dvipsnames]{xcolor}        
\usepackage{amsmath,amssymb,amsthm}
\allowdisplaybreaks
\usepackage{graphicx}
\usepackage{placeins}
\usepackage{float}
\usepackage{enumitem}
\usepackage{algorithm,algpseudocode}

\newcommand{\der}{\mathrm{d}}
\newcommand{\iid}{\overset{\text{iid}}{\sim}}
\newcommand{\E}{\mathbb{E}}
\newcommand{\I}{\mathbb{I}}
\newcommand{\R}{\mathbb{R}}
\newcommand{\drift}{\mathcal{S}_h}
\newcommand{\drifti}{\mathcal{S}_h^{(i)}}
\newcommand{\diff}{\mathcal{C}_h}
\newcommand{\diffi}{\mathcal{C}_h^{(i)}}

\newtheorem{assumption}{Assumption}
\newtheorem{remark}{Remark}
\newtheorem{theorem}{Theorem}[section]
\newtheorem{lemma}[theorem]{Lemma}
\newtheorem{proposition}[theorem]{Proposition}
\newtheorem{definition}[theorem]{Definition}
\newcommand{\keywords}[1]{%
	\smallskip\noindent\textbf{Keywords:} #1\par
}

\title{SULA: Robust Langevin Sampling via Drift-Diffusion Saturation}
\author{
Ritik Soni \\
Department of Statistics and Data Science\\
Indian Institute of Technology Kanpur\\
\texttt{ritiks@iitk.ac.in}
\and
Dootika Vats \\
Department of Statistics and Data Science\\
Indian Institute of Technology Kanpur\\
\texttt{dootika@iitk.ac.in}
}
\date{\today}

\begin{document}
	\maketitle
	
	\begin{abstract}
		The unadjusted Langevin algorithm (ULA) and its stochastic-gradient variant, SGLD, are workhorses of modern Bayesian machine learning, but can be fragile when faced with a target distribution possessing a non-Lipschitz drift. In such settings, the conditional expected squared jump distance can grow rapidly, leading to numerical instability or the need for prohibitively small step-sizes.  We propose  the saturated unadjusted Langevin algorithm (SULA)  that pairs a saturated gradient drift with a gradient-informed state-dependent diffusion, while preserving first-order consistency with the Langevin diffusion. The two coefficients are chosen so that the conditional expected squared jump distance is independent of the gradient. Under suitable conditions, we establish first-order weak convergence to the Langevin diffusion, geometric ergodicity of the discretized Markov chain for a fixed step-size, and an $O(h)$ bound on the invariant-measure bias for smooth test functions with polynomial growth. We also develop a stochastic-gradient variant, SG-SULA, using a gradient rescaling motivated by a Gaussian cumulative distribution function identity under Gaussian gradient noise. Numerical experiments demonstrate that SULA is less sensitive to the choice of step-size than existing Langevin discretizations on the targets considered, and that SG-SULA provides competitive predictive performance in Bayesian inference problems.
		
	\end{abstract}
	
	
	\section{Introduction}
	
	Efficient and reliable sampling from complex probability distributions is a critical component in the workflow of many machine learning and statistical tasks. Traditional Markov chain Monte Carlo (MCMC) algorithms like Metropolis-Hastings require a full sweep of the dataset in every iteration of the sampler, prohibiting scalability. 
	A key driving force for innovation in efficient sampling is the powerful Langevin stochastic differential equation (SDE). Let $\pi(x)$ be the desired target density defined on $\mathbb{R}^d$. The Langevin SDE, $(Y_t)_{t \ge 0}$, is a diffusion process given by
	\begin{equation}\label{eq:langevin}
		\der Y_t = \nabla \log \pi (Y_t) \der t + \sqrt{2}\,\der W_t,
	\end{equation}
	where $(W_t)_{t \geq 0}$ is a $d$-dimensional standard Brownian motion. Under mild conditions, the Langevin SDE has $\pi$ as its unique invariant and limiting distribution \citep{Pavliotis2014,Roberts1996}. Inspired by the Langevin SDE, \cite{Roberts1996} proposed the powerful unadjusted Langevin algorithm (ULA) that employs an Euler-Maruyama discretization of \eqref{eq:langevin}, where for $h > 0$
	\begin{equation}
		\label{eq:ula}
		X_{t+1} = X_t + h\, \nabla \log \pi (X_t)  + \sqrt{2 h} Z_t\,,
	\end{equation}
	where $Z_t \iid N(0,\I_d)$. Under suitable conditions, the process $\{X_t\}_{t \geq 0}$ has been shown to converge to $\pi_{h}^{\text{ULA}}$, an approximation of $\pi$, and as $h \to 0$, $\pi_h^{\text{ULA}} \to \pi$ in a suitable metric \citep{Roberts1996}. ULA and its stochastic gradient variants \citep{welling2011bayesian} are among the most popular sampling algorithms employed in modern machine learning \citep{ahn2012bayesian, bardenet2014towards, duchi2011adaptive, li2016preconditioned, salakhutdinov2008bayesian, welling2011bayesian}. However, ULA requires careful tuning of $h$ and standard convergence guarantees often require the gradient of $\log\pi$ to be globally Lipschitz.  \cite{skew-symmetric} highlight that for ULA 
	$$
	\E\left[ \|X_{t+1} - X_t\|^2 \mid X_t = x\right] = h^2 \left\|\nabla \log \pi(x) \right\|^2 + 2dh\,.
	$$
	When $\nabla \log \pi$ has a large local Lipschitz constant or if it is superlinear in parts of the state-space, the expected squared jump distance of the sampler explodes. This forces smaller values of $h$, which then requires larger run lengths, increasing overall computational burden. To address this source of numerical instability, we present the saturated unadjusted Langevin algorithm (SULA).
	
	\subsection{Related Work}
	
	Some attempts have been made to address the lack of robustness of ULA and SGLD. Tamed Euler methods control superlinear drift growth in numerical approximations of SDEs \citep{hutzenthaler2012strong,sabanis2013note,Sabanis2016}. In the sampling context, \citet{brosse2019tamed} study the tamed unadjusted
	Langevin algorithm (TULA), which replaces $\nabla\log\pi(x)$ in \eqref{eq:ula}
	with $\nabla\log\pi(x)/(1+h\|\nabla\log\pi(x)\|)$, and its coordinate-wise variant, which replaces each component $\nabla_i\log\pi(x)$ with
	$\nabla_i\log\pi(x)/(1+h|\nabla_i\log\pi(x)|)$, for $i=1,\dots,d$. However, the scheme can exhibit large errors for coarse discretizations \citep{skew-symmetric}, as we also observe in our numerical experiments. Further developments provide non-asymptotic convergence guarantees for tamed Langevin methods under a log-Sobolev inequality \citep{lytras2025taming} and for stochastic gradient variants with superlinearly growing, possibly discontinuous gradients \citep{lim2024nonasymptotic}. The Barker proposal of \citet{livingstone2022barker} provides a gradient-based Metropolis-Hastings method designed to improve robustness to tuning and target heterogeneity.
	The unadjusted Barker algorithm (UBA) \citep{skew-symmetric,mauri2024robust} introduces skew-symmetric noise increments whose asymmetry encodes the drift direction. This approach provides robustness and weak convergence under suitable assumptions that allow non-globally Lipschitz $\nabla\log\pi$, but requires an additional Bernoulli flip per coordinate at every iteration.

	In addition to the foundational work on TULA \citep{brosse2019tamed} and UBA \citep{skew-symmetric}, several alternative schemes aim to explore the target distribution efficiently. A popular approach uses Riemannian manifold SDEs with state-dependent diffusion matrices. \citet{girolami2011riemann} were the first to consider Euler-Maruyama discretizations of Riemannian manifold SDEs with further explorations by \citet{ma2015complete}, and extensions to sampling from a simplex by \citet{patterson2013stochastic}. A comprehensive review is provided in \citet{nemeth2021stochastic}. Adaptive and preconditioned variants of SGLD \citep{ahn2012bayesian,li2016preconditioned} address local geometry, which is orthogonal to the gradient-magnitude robustness that SULA targets; the two mechanisms can in principle be combined. In these approaches, state-dependent noise modifies the continuous-time SDE, whereas SULA introduces it only at the discretization level, preserving the target Langevin SDE in \eqref{eq:langevin}.
	
	\subsection{Contributions}
	
	We present SULA, a novel discretization of the SDE in \eqref{eq:langevin}, that is designed to improve robustness to the choice of step-size without modifying the target dynamics 
	or requiring additional sources of randomness. Our contributions are:
	\begin{itemize}[leftmargin=10pt]
		\item \textbf{Variance balancing via drift-diffusion saturation.}  We introduce a saturated drift $\drift(x)$ paired with a   gradient-informed diffusion coefficient $\diff(x)$, chosen so that the expected one-step squared update is always $2dh$. This \emph{variance balancing} property makes the conditional expected squared jump distance independent of the gradient magnitude, yielding robustness to the choice of step-size. We focus on two variants of SULA, $\pi/4-$SULA which reproduces UBA-like behavior in a Gaussian-only framework and $1/2$-SULA which mimics ULA for small gradients and produces predominantly drift-driven updates for large gradients.
		
		\item \textbf{Theoretical guarantees.} We establish three convergence results for SULA under conditions that accommodate certain non-globally Lipschitz gradients. First, we show that as $h \to 0$, SULA converges to \eqref{eq:langevin} in $O(h)$, a rate similar to the ULA algorithm. In practice, the step-size $h$ is never zero; we provide sufficient conditions to show that for any fixed $h\in (0,1)$, the process in \eqref{eq:sula} converges in total variation distance to an invariant distribution $\pi_h$ at a geometric rate. Together with the existence of moments of all orders under $\pi_h$, geometric ergodicity guarantees Markov chain central limit theorems with finite asymptotic variance for ergodic averages of functions with at most polynomial growth \citep{jones2004markov}.
		Finally, we show that under suitable conditions, as $h \to 0$, the difference between expectations of functions $f\in C_P^\infty(\mathbb{R}^d)$ under $\pi_h$ and $\pi$ is $O(h)$, ensuring a bound on the bias induced by the discretization. 
		\item \textbf{Stochastic gradient extension (SG-SULA).} 
		Because the saturated drift $S_h(x)$ is nonlinear, naive replacement 
		of $\nabla \log \pi$ with a stochastic gradient introduces a 
		systematic bias in the drift. We show that  stochastic gradients 
		effectively shrink the true saturated drift by a multiplicative factor, and develop SG-SULA, which motivates a rescaling intended to mitigate this bias. Under Gaussian gradient noise, Proposition~3.2 establishes an exact
		Gaussian CDF identity that motivates the gradient rescaling used
		in SG-SULA making it suitable for large-scale Bayesian inference.
		\item \textbf{Empirical validation.} We evaluate SULA and SG-SULA 
		across three settings of increasing complexity: (i) \emph{Known targets:} A heterogeneous $100$-dimensional Gaussian and a super-Gaussian distribution where ULA is known to break down; 
		(ii) \emph{Bayesian neural encoding:} A count-response Bayesian model fit to spike train recordings, representing a realistic  inference problem with correlated  covariates and sparse event counts; (iii) \emph{Bayesian matrix factorization:} SG-SULA is evaluated on the MovieLens dataset against SGLD and the corrected stochastic gradient Barker dynamics (SGBD), comparing predictive performance across the sampling methods.
	\end{itemize}
	
	\section{Robust Langevin Sampling via SULA}
	Consider the task of sampling from a distribution with density $\pi$ defined on $\R^d$. 
	The SULA update is
	\begin{equation}
		\label{eq:sula}
		X_{t+1} = X_t + h\, \drift (X_t) + \sqrt{2h} \,\diff(X_t) Z_t\,,
	\end{equation}
	where for $c > 1/2$ and $i = 1, 2, \dots, d$, the $i$th component of $\drift$ and the $i$th diagonal entry of the diagonal matrix $\diff$ are
	\begin{equation}
		\label{eq:SULA_drift}
		\drift^{(i)}(x) = \dfrac{\nabla_i \log \pi(x)}{\sqrt{1 + c h  \left(\nabla_i \log \pi(x)\right)^2 } } \qquad
		\text{ and } \quad \diffi(x) = \sqrt{1 - \dfrac{h}{2} \drifti(x)^2}\,.
	\end{equation}
	Setting $c = \pi/4$ yields conditional first and second moments similar to those of the UBA algorithm. For fixed $c > 1/2$, as $\nabla_i\log\pi(x)\to\pm\infty$, the saturated drift $S_h^{(i)}(x)$ approaches $\pm1/\sqrt{ch}$, while $C_h^{(i)}(x)$ approaches $\sqrt{1-1/(2c)}$. Consequently, choosing $c$ close to $1/2$ results in predominantly drift-driven updates when the gradient magnitude is large. We refer to the scheme with $c=0.51$ as $1/2$-SULA.
	We note that the drift-driven  regime of $1/2-$SULA is itself $h$-dependent. For any fixed $x$, $\drift(x) \to \nabla \log \pi(x)$ and $\diff(x) \to \mathbb{I}_d$ as $h \to 0$, so $1/2-$SULA recovers Langevin behavior in the small-$h$ limit. Our numerical experiments will highlight the benefit of choosing $c \approx 1/2$  and our results hold for all $c > 1/2$. By design, the drift term $\drifti(x) \approx \nabla_i \log \pi(x)$ for small $h$, similar to the TULA sampler. However, because of the diffusion coefficient introduced in the discretization, we obtain
	$$
	\E \left[ (X_{t+1}^{(i)} - X_t^{(i)} )^2 |\; X_t\right]
	= 2h.
	$$
	We now introduce notation used throughout the paper. We write $\mathbb{E}_x$ for expectation with respect to the process initialized at $x$. The total variation norm of a measure $\mu$ is denoted by $\|\mu\|_{\mathrm{TV}}$.  For $l \ge 1$, we denote by $C_P^l(\mathbb{R}^d)$ the class of $l$-times continuously differentiable functions whose derivatives up to order $l$ have at most polynomial growth. Throughout, $(Y_t)_{t \ge 0}$ denotes the solution of the Langevin SDE \eqref{eq:langevin} with invariant measure $\pi$, and $\{X_t\}_{t \ge 0}$ denotes the Markov chain generated by the SULA update. 
	
	\subsection{Theoretical Results}
	We now show that SULA converges weakly to the Langevin diffusion \eqref{eq:langevin} as $h \to 0$. To this end, we impose the following assumptions on the target density $\pi$. 
	
	\begin{assumption}\label{ass:main}
		For every $i=1,\dots,d$, assume that $\nabla_i\log\pi\in C_P^\infty(\mathbb{R}^d)$. In addition, the following hold:
		\begin{enumerate}[label=(\alph*), ref=\theassumption(\alph*), leftmargin=20pt]
			\item \label{assum1}
			There exist $R_0,s_0,\kappa_0>0$ such that, for every
			$x\in\mathbb{R}^d$ with $\|x\|\ge R_0$ and every $s\ge 0$,
			\begin{equation*}
				\sum_{i:\,|\nabla_i\log\pi(x)|>s}
				x_i\,\operatorname{sign}\bigl(\nabla_i\log\pi(x)\bigr)
				\le
				\begin{cases}
					-\kappa_0\|x\|, & 0\le s\le s_0,\\
					0,              & s>s_0.
				\end{cases}
			\end{equation*}
			
			\item \label{assum:LC}	There exist $\alpha>0$ and $\beta\ge 0$ such that for all
			$x,y\in\mathbb{R}^d$,
			\begin{equation*}
				\langle \nabla\log \pi(x) - \nabla\log \pi(y),\, x-y \rangle
				\le \beta - \alpha \|x-y\|^2 .
			\end{equation*}
			
			\item \label{assum:HP}
			There exists a constant $C>0$ such that
			$
			\nabla^2 \log \pi(x) \preceq C I_d \text{ for all } x\in\mathbb{R}^d.
			$
			
			\item \label{assum:PG}
			There exist constants $m,L>0$ such that for all $x,y\in\mathbb{R}^d$,
			\begin{equation*}
				\|\nabla\log \pi(x) - \nabla\log \pi(y)\|
				\le 
				L\bigl(1+\|x\|^{m} + \|y\|^{m}\bigr)\,\|x-y\|.
			\end{equation*}
		\end{enumerate}
	\end{assumption}
	Assumption~\ref{assum1} allows individual gradient components to point outwards while ensuring that the saturated drift points inwards in the tails (see Lemma~\ref{lem:coercivity}). In particular, it accommodates correlated Gaussian targets whenever their covariance matrix is strictly diagonally dominant (see Proposition \ref{prop:sufficient}). Assumption~\ref{assum:LC} is a dissipativity condition with an additive constant, allowing superlinear drifts. Assumption~\ref{assum:HP} is a smoothness condition imposing a uniform upper bound on the Hessian of $\log \pi$, allowing
	unbounded negative curvature. This assumption is immediately satisfied by all weakly log-concave distributions. Finally, Assumption~\ref{assum:PG} imposes a polynomial growth condition on the gradient and is a technical condition employed in the proofs. Although this condition is not satisfied in Section~\ref{sec:neural}, our empirical results illustrate the applicability of SULA beyond the scope of the theoretical assumptions. 
	
	We aim to study the weak convergence of SULA to the underlying Langevin SDE. 
	Under Assumption~\ref{ass:main}, we now state the main weak convergence result.
	\begin{theorem}\label{thm:weak-true}
		Under Assumption~\ref{ass:main}, for any $T>0$ and any $f \in C_P^4(\mathbb{R}^d)$, there exists a function $\mathcal{V} \in C_P^\infty(\mathbb{R}^d)$ such that for all $h \in (0,1)$ and all integers $k$ with $kh \le T$,
		\begin{equation}
			\left| \mathbb{E}_x\big[f(X_k)\big] - \mathbb{E}_x\big[f(Y_{kh})\big] \right|
			\le \mathcal{V}(x)h,
		\end{equation}
		where $(Y_t)_{t \ge 0}$ solves the Langevin SDE \eqref{eq:langevin} with initial condition $Y_0 = x$.
	\end{theorem}
	
	\begin{proof}[Proof Sketch.]
		The result follows by verifying the conditions in Theorem 2.2.1 of \cite{milstein2004stochastic}. First, Proposition~\ref{prop:mombound} ensures that Assumption (d) in Theorem 2.2.1 of \cite{milstein2004stochastic} is satisfied. Next, Proposition~\ref{lmm:Weakexp} establishes Equation (2.2.7) with $p=1$. In addition, Theorem~\ref{thm:DiffusionSDE} guarantees that this condition holds for any test function.
		Combining these results, all the assumptions in Theorem 2.2.1 of \cite{milstein2004stochastic} are verified, which completes the proof.
	\end{proof}
	
	Theorem~\ref{thm:weak-true} ensures that as $h \to 0$ the discretized process in \eqref{eq:sula} converges to the Langevin SDE in \eqref{eq:langevin} in the weak sense. In practice, however, letting $h \to 0$ forces more iterations of the sampler to achieve the same overall time evolution, $T$, thus increasing computational load. Thus, it is critical to study  the convergence behavior of the sampler for a fixed $h$. To this end, we show that the SULA chain is geometrically ergodic for a unique invariant distribution, $\pi_h$. 
	
	\begin{theorem}\label{thm:geom-ergodicity}
		Under Assumption~\ref{assum1}, for each fixed $h\in(0,1)$,
		the Markov chain generated by SULA admits an invariant measure
		$\pi_h$ with finite moments of all orders. Furthermore, denoting by $P_{h}$ the associated Markov kernel, there exist constants
		$C>0$, $\lambda>0, a>0$, and a function $W_a:\mathbb{R}^d \to [1,\infty)$ such that for all
		$n\in\mathbb{N}$ and all $x\in\mathbb{R}^d$,
		\begin{equation}
			\label{eq:geom-ergodicity}
			\bigl\| P_{h}^n(x,\cdot) - \pi_{h} \bigr\|_{\mathrm{TV}}
			\le C\,W_a(x)\,e^{-\lambda n}\,.
		\end{equation}
	\end{theorem}
	
	\begin{proof}[Proof Sketch]
		Following \cite{Lelievre2018} and \cite{Roberts1996}, in Proposition~\ref{prop:irreducibility} we first establish that the kernel $P_h$ is irreducible with respect to the Lebesgue measure, aperiodic, and admits a minorization condition.
		Next, Proposition \ref{prop:RoTULA-drift} establishes a drift condition for the SULA scheme with a specific choice of a drift function $W_a(x)$. From Theorem 2.25 of \cite{Lelievre2018}, 
		the sufficient conditions for geometric ergodicity are satisfied. 
	\end{proof}
	
	Theorem~\ref{thm:geom-ergodicity} establishes two key properties that immediately guarantee the existence of Markov chain central limit theorems for ergodic averages \citep{jones2004markov}. This immediately implies that estimators of expectations of functions have finite variance. However, the discretization induces an asymptotic bias since $\pi_h \ne \pi$ in general. The next result establishes a vanishing bias for such expectations as $h \to 0$. First we state an assumption.
	
	\begin{assumption}\label{assum:bias-result}
		The following conditions are assumed on $\pi$ and $(Y_t)_{t \geq 0}$.
		\begin{enumerate}[label=(\alph*), ref=\theassumption(\alph*), leftmargin=20pt]
			
			\item \label{assum:smooth} The target density $\pi \in C^4(\mathbb{R}^d)$,  $\pi(x) >0$ for all $x \in \mathbb{R}^d$, and $\log \pi \in C_P^4(\mathbb{R}^d)$.
			\item \label{assum:moments}The measure $\pi$ has moments of all orders.
			\item \label{assum:sde} The diffusion process $(Y_t)_{t \ge 0}$ in the Langevin SDE \eqref{eq:langevin} is time reversible and has $\pi$ as its unique invariant and limiting measure. 
		\end{enumerate}
	\end{assumption}
	
	\noindent Now, we provide the following result that establishes the discretization error.
	\begin{theorem}\label{thm:bias}
		Under Assumptions~\ref{ass:main} and \ref{assum:bias-result}, there exists $\phi\in C_P^\infty(\mathbb{R}^d)$ with $\int\phi\,d\pi=0$ such that, for every $f\in C_P^\infty(\mathbb{R}^d)$, there exists a constant $K>0$, depending on $f$, such that for all $h\in(0,1)$,
		$$
		\int f\,d\pi_h = \int f\,d\pi + h\int f\phi\,d\pi + h^2R_h^f, \qquad |R_h^f|\leq K.
		$$
	\end{theorem}
	\begin{proof}[Proof Sketch] 
		We verify the assumptions in Theorem 3.3 of \cite{Lelievre2018} with the help of Proposition \ref{lmm:Weakexp}, which leads to the required bias result.
	\end{proof}
	
	\section{Stochastic Gradient SULA}
	Stochastic gradient MCMC methods like SGLD have become a standard tool for Bayesian inference in modern machine learning when complete gradient computation is infeasible. However, SGLD suffers from the same fragility issues as ULA. As a result, there is a growing interest in developing sampling schemes that remain stable under noisy gradients \citep{chen2014stochastic,mauri2024robust,patterson2013stochastic,  teh2016consistency}. In this direction, we propose a stochastic gradient variant of the SULA algorithm, called SG-SULA. For $k = 1, \dots, N$, let $\pi_k$ denote the target density contributed by the $k$th observation with the log-prior contribution distributed.  Here $N$ is the number of samples in the data. For each coordinate $i=1, \dots, d$, the stochastic gradient estimator of $\nabla_i \log \pi$ is
	\begin{equation}
		\widehat{\nabla}_i \log \pi(x)
		=
		\frac{N}{n} \sum_{j \in S_n} \nabla_i \log \pi_j(x),
	\end{equation}
	where $S_n \subset \{1,\dots,N\}$ is a subset of size $n \ll N$, sampled uniformly at random. 
	
	Critically, although the stochastic gradient estimator is unbiased at the gradient level, the nonlinear structure of $\mathcal{S}_{h}$ introduces bias in the stochastic gradient variant. The SGBD of \cite{mauri2024robust} faces a similar challenge where it is resolved by imposing additional assumptions on the gradient noise. Specifically, under Gaussianity of the gradient noise, \cite{mauri2024robust} show that a corrected version of SGBD yields \emph{almost} unbiased one-step updates. Following a similar approach, we use a Gaussian CDF identity to
	motivate a gradient rescaling intended to mitigate the shrinkage
	introduced by stochastic gradients. We define for each $i=1,\dots,d,$ and $c >1/2,$
	\begin{equation}
		\widehat{\mathcal{S}}^{(i)}_{h}(x, \alpha)
		: =
		\frac{\alpha \widehat{\nabla}_i \log \pi(x)}{
			\sqrt{1 + c\,h (\alpha  \widehat{\nabla}_i\log\pi(x) )^2}}, \qquad \widehat{\mathcal{C}}^{(i)}_{h}(x, \alpha)
		=
		\sqrt{1 - \frac{h}{2} \widehat{\mathcal{S}}^{(i)}_{h}(x, \alpha)^2 }.\end{equation}
	The drift $\widehat{\mathcal{S}}^{(i)}_{h}(x, \alpha)$ multiplies the gradient by a factor of $\alpha$ so that $\widehat{\mathcal{S}}^{(i)}_{h}(x, 1)$ is the same as \eqref{eq:SULA_drift} with stochastic gradients. We first assume the stochastic gradient is of the form
	$
	\widehat{\nabla} \log \pi(x) = \nabla \log \pi(x) + \eta_t,
	$
	where $\eta_t = (\eta^{(1)}_t, \eta^{(2)}_t,\ldots, \eta^{(d)}_t)$ represents the stochastic noise. The following result establishes that the stochastic gradient  reduces the magnitude in conditional expectation of the saturated drift when the stochastic noise is symmetric. The proof of the following proposition is in the appendix.

	\begin{proposition}\label{sgprop1}
		Assuming $\eta^{(i)}_t \stackrel{d}{=} -\eta^{(i)}_t$ conditionally on $X_t$, for each $i = 1,\dots,d$, we obtain a shrinkage bound
		\begin{equation}
			\left|\mathbb{E}\left[\widehat{\mathcal{S}}^{(i)}_{h}(X_t, 1) \mid X_t\right]\right|
			\le
			\left|\mathcal{S}^{(i)}_{h}(X_t)\right|.
		\end{equation}
	\end{proposition}
	Proposition~\ref{sgprop1} shows that stochastic gradients shrink the original drift, and therefore suggest a correction by rescaling. For $Z\sim\mathcal{N}(0,1)$, the drift in \eqref{eq:SULA_drift} admits the representation
	\begin{equation}\label{eq:integral-form}
		\drifti(x)=\sqrt{\frac{2\pi}{ch}}\,
		\E\bigl[Z\Phi\bigl(\sqrt{ch}\,\nabla_i
		\log\pi(x)Z\bigr)\big].
	\end{equation}
	Following the strategy of \citet{mauri2024robust}, the next proposition establishes a rescaling that corrects the Gaussian CDF term after averaging over the gradient noise for each fixed $z$. The proof is in the appendix.
	\begin{proposition}\label{sgprop2}
		Let $\eta_t^{(i)}\sim\mathcal{N}(0,\tau^2)$ conditionally on $X_t$ and $\tau>0$. For fixed $z\in\mathbb{R}$ such that $|z|<1/(\tau\sqrt{ch})$, define $\alpha_{z,h}^{\tau}=(1-chz^2\tau^2)^{-1/2}$. Then
		\begin{equation}
			\mathbb{E}\left[\Phi\left(\alpha_{z,h}^{\tau}\sqrt{ch}\,\widehat{\nabla}_i\log\pi(X_t)z\right)\;\middle|\;X_t\right]=\Phi\left(\sqrt{ch}\,\nabla_i\log\pi(X_t)z\right).
		\end{equation}
	\end{proposition}
	The Gaussian assumption on the stochastic gradients is certainly limiting but may often hold true when batch size, $n$, is large enough. Proposition~\ref{sgprop2} gives an exact correction of the Gaussian CDF term for each fixed $z$ whenever $|z| < 1/(\tau \sqrt{c h}).$  However, outside this regime, the correction factor $\alpha_{z,h}^\tau$ is not well-defined. Thus, we define $\tilde{\alpha}_{z,h}^\tau = \alpha_{z,h}^\tau$ when $|z| < 1/(\tau \sqrt{ch})$ and $1$ otherwise. Through the representation in \eqref{eq:integral-form}, this
	provides a motivation for rescaling the stochastic
	gradient in SULA. We apply this rescaling to both
	coefficients in \eqref{eq:SULA_drift} to construct
	SG-SULA. With this, we present the SG-SULA algorithm in Algorithm~\ref{csg-algo}, where $\texttt{var}_{S_n}^{(i)}$ denotes the sample variance of the $i$th component of the stochastic gradients in batch $S_n$.
	
	We emphasize that Proposition 3.2 establishes an exact correction of the Gaussian CDF term for a fixed admissible $z$ and known $\tau$, rather than an exact correction of the saturated drift in Algorithm 1. The resulting rescaling should therefore be viewed as a heuristic motivated by the integral representation in \eqref{eq:integral-form}.  Since the rescaling factor depends on the Gaussian variable used in the update, the exact conditional second-moment identity of SULA need not hold for SG-SULA. We investigate the empirical performance of the resulting scheme in the next section.
	
	
	\begin{algorithm}[t]
		\caption{Stochastic Gradient SULA (SG-SULA)}
		\label{csg-algo}
		\begin{algorithmic}[1]
			\State \textbf{Input:} $X_0 \in \mathbb{R}^d$, step-size $h>0$,
			\Statex \hspace{\algorithmicindent} $\beta\in(0,1)$,
			initial standard deviations $\{\widehat{\tau}^{(i)}_0\}_{i=1}^d$
			\For{$t=0,1,2,\dots$}
			\State Sample $S_n \subset \{1,\dots,N\}$ of size $n$
			uniformly at random
			\For{$i=1,\dots,d$} \Comment{Parallelizable}
			\State Compute $\widehat{\nabla}_i\log\pi(X_t)$ using $S_n$
			\State Update the gradient-noise standard deviation
			\[
			\widehat{\tau}^{(i)}_{t+1}
			= (1-\beta)\widehat{\tau}^{(i)}_t
			+ \beta\,\frac{N}{\sqrt{n}}
			\sqrt{\frac{N-n}{N-1}}\,
			\sqrt{\texttt{var}_{S_n}^{(i)}}.
			\]
			\State Draw $Z_t^{(i)}=z^{(i)}\sim\mathcal{N}(0,1)$
			independently of the minibatch and other Gaussian draws
			\State Compute $\tilde{\alpha}_{z^{(i)},h}^{\tau}$
			with $\tau=\widehat{\tau}_{t+1}^{(i)}$
			\State Compute modified drift
			$\widehat{\mathcal{S}}_h^{(i)}
			(X_t,\tilde{\alpha}_{z^{(i)},h}^{\tau})$
			and modified diffusion
			$\widehat{\mathcal{C}}_h^{(i)}
			(X_t,\tilde{\alpha}_{z^{(i)},h}^{\tau})$
			\State Update
			\[
			\begin{aligned}
				X_{t+1}^{(i)}
				&= X_t^{(i)}
				+ h\,\widehat{\mathcal{S}}_h^{(i)}
				(X_t,\tilde{\alpha}_{z^{(i)},h}^{\tau})+ \sqrt{2h}\,
				\widehat{\mathcal{C}}_h^{(i)}
				(X_t,\tilde{\alpha}_{z^{(i)},h}^{\tau})Z_t^{(i)}.
			\end{aligned}
			\]
			\EndFor
			\EndFor
		\end{algorithmic}
	\end{algorithm}

	\section{Numerical Experiments}
	
	We present three sets of examples of increasing complexity. The first set are implementations of SULA on targets that provide proof of concept and highlight utility of the method. Next, we present a Bayesian study of neural encoding where the model possesses a superlinear drift. Standard algorithms struggle to effectively employ gradients here. We implement both unadjusted and stochastic gradient methods. Our final example is that of a Bayesian probabilistic matrix factorization studying the MovieLens dataset. The size of this dataset is prohibitively large so that we only implement stochastic gradient methods here. Throughout, we compare SULA and SG-SULA to ULA, TULA, UBA and their stochastic gradient counterparts; to ensure reproducibility, code for SGBD was used directly from the open access repository shared in \cite{mauri2024robust}. All experiments were conducted on a Linux machine with an Intel Xeon Gold 6226R CPU (16 cores dual socket) and 256GB RAM, running Ubuntu 22.04.
	
	\subsection{Toy Examples: Gaussian and Super-Gaussian}
	
	We first demonstrate robustness to discretization of our proposed SULA algorithm on two targets: (i) a heterogeneous 100-dimensional Gaussian target with marginal variances on a grid from 10 to 0.1, and (ii) a 100-dimensional super-Gaussian distribution with independent components $\propto e^{-x^4}$. The Gaussian target is designed to have heterogeneously scaled components. Consequently, choosing a step-size $h$ without any preconditioning can be challenging, with preconditioning often  being unreliable \citep{hird2025quantifying} as well as computationally burdensome. Similarly, the super-Gaussian distribution has a non-Lipschitz log-gradient, but one that is one-sided Lipschitz and exhibits polynomial growth, satisfying Assumption~\ref{ass:main}.
	
	\begin{figure*}[t]
		\centering
		\includegraphics[width=0.35\linewidth]{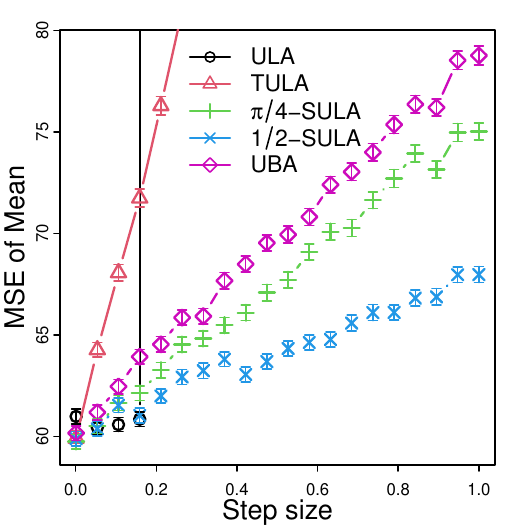}
		\includegraphics[width=0.35\linewidth]{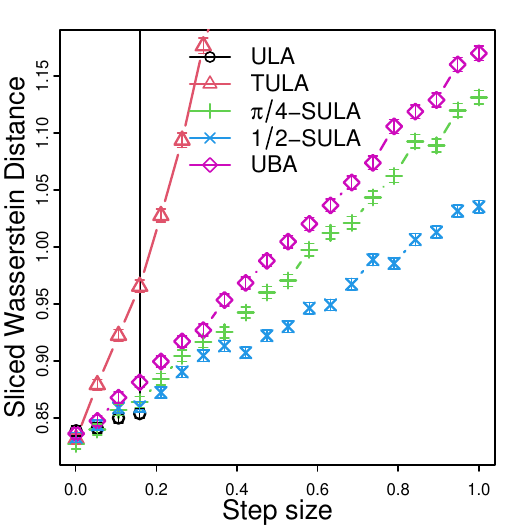}
		\includegraphics[width=0.35\linewidth]{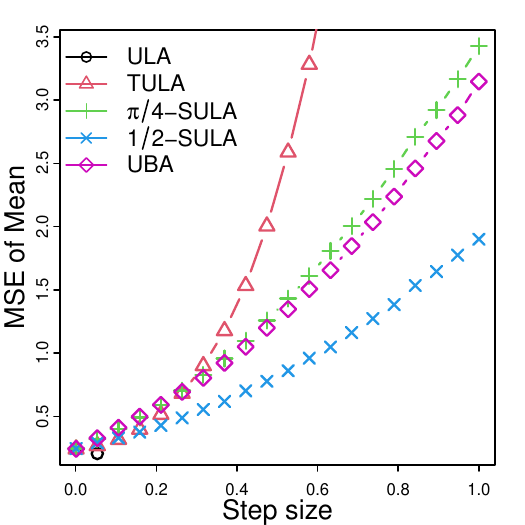}
		\includegraphics[width=0.35\linewidth]{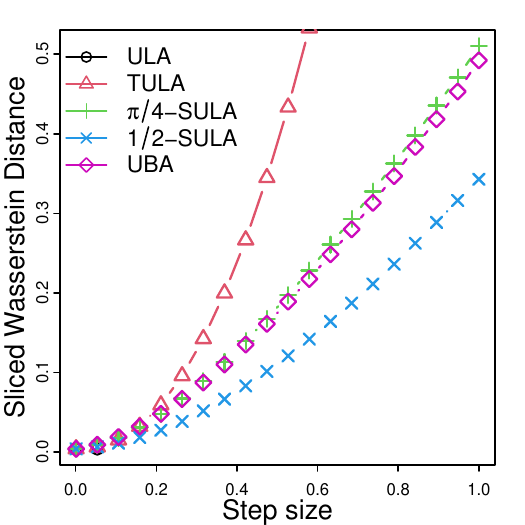}        
		\caption{(Top row): Mean squared error for estimating the mean and estimated sliced Wasserstein distance from the truth for the Gaussian example. (Bottom row): Mean squared error for estimating the mean and estimated sliced Wasserstein distance from the truth for the super-Gaussian example. Standard errors over 1000 replications are also plotted.}
		\label{fig:toy}
	\end{figure*}
	
	We ran 1000 replications of all algorithms with SULA being run for $c = \pi/4$ and $c = 0.51$. We call these $\pi/4-$SULA and $1/2-$SULA, respectively. The simulation was run for varying step-sizes, keeping the process runtime fixed at $T = 100$. Over repeated simulations, we compare the mean squared error (MSE) in the sampling mean and the sliced Wasserstein distance \citep{bonneel2015sliced,nadjahi2020statistical} between the obtained samples and independent samples obtained from the truth. Figure~\ref{fig:toy} summarizes the results. Noticeably, ULA  breaks down, escaping from the plot after small step-sizes in both examples. TULA exhibits larger empirical errors than SULA and UBA, although it does not break down as dramatically as ULA. On the other hand, SULA and UBA exhibit more robust behavior with $\pi/4-$SULA and UBA being similar (as expected). In both examples, the $1/2-$SULA sampler exhibits smaller empirical errors across the range of step-sizes considered. 
	
	\subsection{Bayesian Study of Neural Encoding}
	\label{sec:neural}
	We evaluate our SULA scheme on a neural encoding task using spike train recordings within a count response Bayesian model. Such models are widely used in computational neuroscience to characterize how external stimuli and spike history influence neural firing rates. This setting presents a challenging inference problem due to high-dimensional correlated covariates and sparse event counts. A Bayesian formulation enables uncertainty quantification over stimulus filters and firing rates, which is essential for scientific interpretability. Consequently, this dataset serves as a realistic benchmark for scalable inference in non-conjugate, high-dimensional probabilistic models.  
	
	We use a spike train dataset reported in \cite{d2023efficient}, where calcium imaging is used to record neuronal activity in response to external stimuli. The dataset contains $N = 920$ observations, each representing the number of activation events for a neuron. Associated with each observation is a vector of $p = 23$ covariates, denoted by $x_i \in \mathbb{R}^{23}$, which captures relevant features. To model the count data, we adopt a Bayesian Poisson regression framework proposed in \cite{d2023efficient}. For each $i = 1, \dots, N$, we have
	$$
	Y_i \mid \beta \sim \mathrm{Poisson}\big(\exp(x_i^\top \beta )\big),
	$$
	where $\beta  \in \mathbb{R}^{23}$ is the unknown parameter vector with Gaussian prior, that is,
	$
	\beta  \sim \mathcal{N}(0, 2 I_{23}),
	$
	resulting in an intractable posterior. To obtain a reliable benchmark, we generate reference samples using a carefully tuned Tamed-MALA algorithm. 
	
	\begin{figure}
		\centering
		\includegraphics[width=0.60\linewidth]{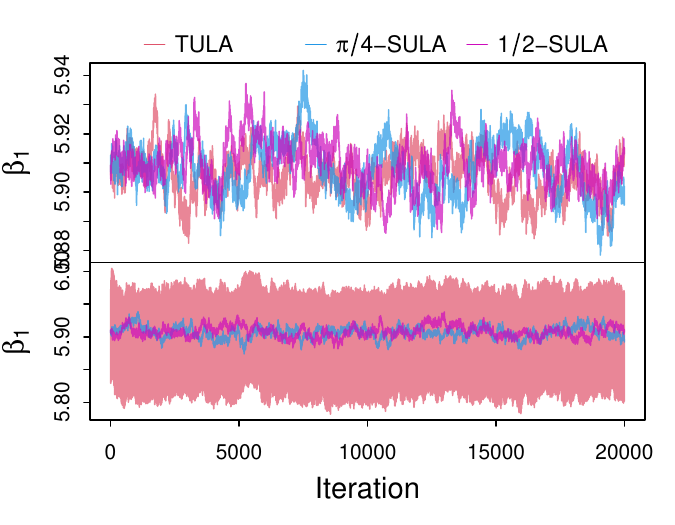}
		
		\includegraphics[width=0.35\linewidth]{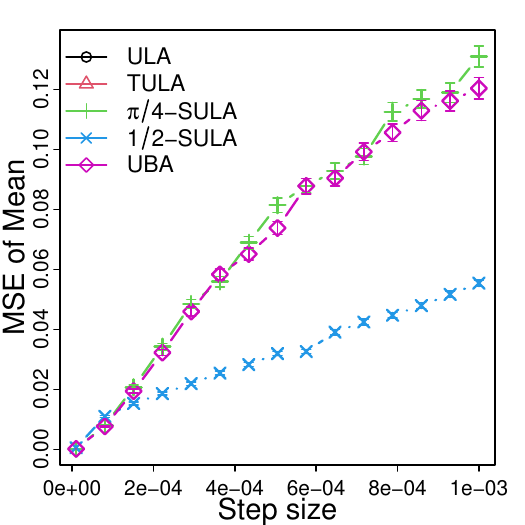}
		\includegraphics[width=0.35\linewidth]{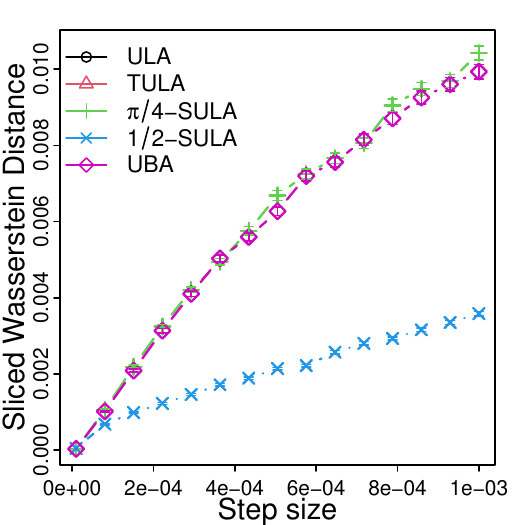}
		\caption{(Top) Trace plot of the first regression coefficient for $h = 8\times 10^{-7}$ (upper panel) and $h = 1.6\times 10^{-6}$ (lower panel). (Bottom right) Mean squared error for the mean and (Bottom left) estimated sliced Wasserstein distance, both as a function of increasing step-size. Standard errors over 1000 replications are also plotted.}
		\label{fig:neural}
	\end{figure}
	First, we demonstrate that despite careful taming of the gradient, TULA struggles in this model. We set $h = 8 \times 10^{-7}$ and then subsequently double $h$ to obtain $h = 1.6 \times 10^{-6}$. Trace plots for the first component are given in Figure~\ref{fig:neural} (left). Noticeably, SULA remains robust to doubling of the step-size due to the specific variance balancing employed, however TULA is drastically affected by this change of discretization. We further compare SULA to UBA using replicated experiments. All algorithms are initialized at the maximum likelihood estimate and run for $T = 0.1$ time units for various choices of step-size. Performance is evaluated over $1000$ independent replications. We report the  MSE of the estimated posterior mean, along with the sliced Wasserstein distance between the empirical sample distributions and those produced by the reference sampler, averaged across independent runs. Figure \ref{fig:neural} (right) reports the behavior which closely mirrors that seen in the toy example in Figure \ref{fig:toy}. In particular, ULA becomes unstable as the step-size increases, leading to divergence. TULA exhibits improved numerical stability but incurs larger empirical errors at larger step-sizes. In contrast, both SULA and UBA demonstrate significantly more stable behavior, with $\pi/4$-SULA closely matching UBA. Notably, the $1/2$-SULA variant consistently achieves the best performance.  
	
	We also implement stochastic gradient algorithms with $n = 100$ and compare the sliced Wasserstein 	distance to the reference sample distribution as well as the effective sample size (ESS) per second \citep{vats2019multivariate}. We do not implement SGLD as SGLD breaks down at step-sizes larger than $10^{-7}$. It is well known that conditional unbiasedness of the stochastic gradient update does not translate to convergence of stochastic gradient methods to $\pi_h$ for a fixed $h$ \citep{teh2016consistency}, since the noise in the stochastic gradient interacts non-trivially with the Gaussian update. Consequently, increasing $h$ does not always increase the bias, therefore, stochastic gradient implementations often use larger step-sizes. Results in Figure~\ref{fig:neural-sg} echo this behavior. Here, $1/2$-SG-SULA exhibits smaller estimated sliced Wasserstein distances over the step-sizes considered while also yielding higher ESS per second. 
	\begin{figure}
		\centering
		\includegraphics[width=0.35\linewidth]{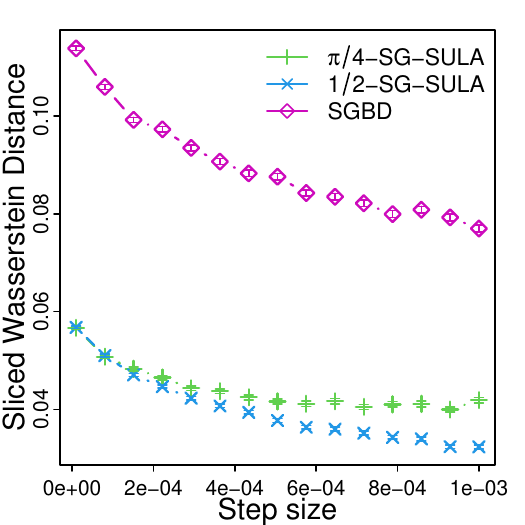}
		\includegraphics[width=0.35\linewidth]{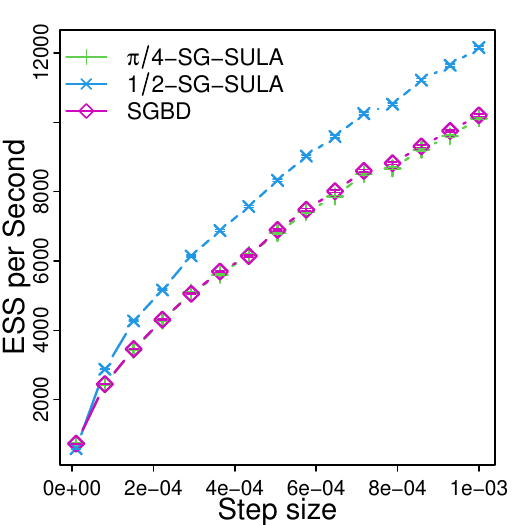}
		\caption{(Left) Estimated sliced Wasserstein distance as a function of step-size. (Right) Effective sample size per second of the algorithms. Standard errors over 1000 replications are also plotted.}
		\label{fig:neural-sg}
	\end{figure}
	
	\subsection{Bayesian Probabilistic Matrix Factorization}
	
	To demonstrate the strength of SULA in a high-dimensional realistic setting, we employ the Bayesian probabilistic matrix factorization (BPMF) model of \cite{salakhutdinov2008bayesian}, which has a high-dimensional latent structure and sparse observed user–item interactions. The MovieLens dataset with $10^5$ ratings in $\{1,2,3,4,5\}$ from 1000 users and 1700 movies clearly exhibits these characteristics, making it a standard benchmark for evaluating scalable sampling methods. We use this setting to compare our proposed SG-SULA scheme with SGLD and the corrected variant of SGBD \cite{mauri2024robust} to assess its empirical performance. Let $R_{ij}$ denote the rating of user $i$ for item $j$. Then for observed entries,
	$$
	R_{ij} \mid U, V \sim \mathcal{N}(U_i^\top V_j, \alpha^{-1}),
	$$
	where $U_i, V_j \in \mathbb{R}^p$ are latent user and item factors. We place Gaussian priors $U_i \sim \mathcal{N}(\mu_U, \Lambda_U^{-1})$ and $V_j \sim \mathcal{N}(\mu_V, \Lambda_V^{-1})$ for $i=1,\dots,U$ and $j=1,\dots,M$, with diagonal precision matrices $\Lambda_U=\mathrm{diag}(\lambda_{U,1},\dots,\lambda_{U,p})$ and $\Lambda_V=\mathrm{diag}(\lambda_{V,1},\dots,\lambda_{V,p})$. The mean parameters follow $\mu_U \sim \mathcal{N}(0, \Lambda_U^{-1})$ and $\mu_V \sim \mathcal{N}(0, \Lambda_V^{-1})$, while the precision components satisfy $\lambda_{U,k}, \lambda_{V,k} \sim \mathrm{Gamma}(\alpha_0,\beta)$ for $k=1,\dots,p$. Similar to \cite{mauri2024robust}, we set $\alpha_0 = 1, \beta = 5$, and $\alpha = 3$. The resulting posterior distribution on the parameters is intractable and sampling or other approximation methods become necessary. 
	\begin{figure}[H]
		\centering
		\includegraphics[width=0.35\linewidth]{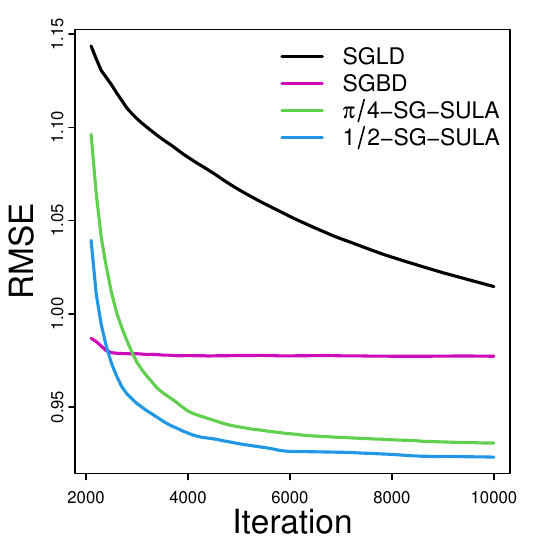}
		\includegraphics[width=0.35\linewidth]{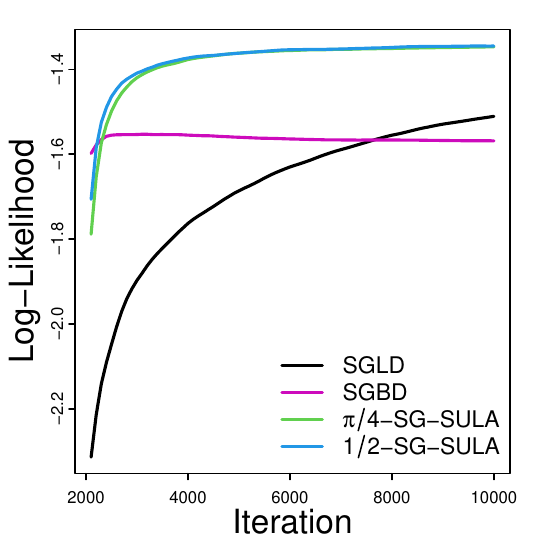}
		\caption{(Left) Root MSE of the predicted user ratings and (right) predictive log-likelihood.}
		\label{fig:bpfm}
	\end{figure}
	
	Experiments are conducted with 20 latent factors. Under this specification, the posterior distribution over $(U, V, \mu_U, \mu_V, \Lambda_U, \Lambda_V)$ is $54{,}080$-dimensional. We use an $80/20$ train-test split, mini-batches of size $800$ and report root MSE of predicted user ratings on the test set as well as predictive log-likelihood of the test set, using running MCMC estimates. For SGLD, we used $h = 0.0078$ as anything larger than that exploded the gradient updates. For SGBD and SG-SULA we set $h = 0.025$. Figure \ref{fig:bpfm} presents the results after a burn-in of 2000 iterations, where evidently both variants of the SG-SULA sampler demonstrate stable and accurate predictive behavior. Noticeably, SG-SULA methods move slower than other methods early on in the exploratory phase, but quite rapidly return improvements in prediction. On the other hand, under the experimental configuration considered, corrected SGBD yields poorer predictive performance than SG-SULA.
	
	\section{Discussion}
	
	We propose SULA, an Euler-Maruyama style discretization of the Langevin SDE that controls the expected squared jump distance independently of the gradient magnitude. This is achieved by employing a saturation function of the gradient in the drift while also using a gradient-informed diffusion term. The two terms are carefully chosen to yield a desired variance balancing.  SULA is particularly effective when the $\nabla \log \pi$ is superlinear or has a large Lipschitz constant. Our theoretical results hold for all $c > 1/2$, and we focus on two specific variants: (i) $\pi/4-$SULA, which reproduces UBA-like first and second moments without requiring the additional Bernoulli randomness of UBA, and (ii) $1/2-$SULA, which mimics ULA when gradients are small and behaves like a descent algorithm when gradients are large, while still recovering the Langevin diffusion in the small-$h$ limit.  Additionally, under Gaussian gradient noise, we derive a correction for the Gaussian CDF term in the integral representation of the saturated drift and use it to motivate the SG-SULA rescaling. Our numerical experiments highlight the robustness of the SULA scheme and predictive performance of SG-SULA.

\begin{paragraph}{Limitations.}
    Our rates are stated only in terms of $h$, with dimension absorbed into the constants. A careful non-asymptotic analysis would highlight this dependence, but due to the nonlinear nature of both the drift and diffusion function, such an analysis is fairly challenging. Further for the reduced bias result of SG-SULA, the Gaussianity assumption on the gradient is limiting and would require large mini-batch sizes particularly for superlinear gradient functions. Finally, both UBA and SULA can suffer from slow drift without an educated starting value. If a warm start is not possible, we first recommend running TULA to obtain a warm start.
    
\end{paragraph}

\begin{paragraph}{Future Work.}
   Several extensions of this work would be critical. A non-asymptotic analysis on the rate of the convergence of the SULA algorithm would highlight the dimension dependence of the algorithm. Extensions to proximal algorithms for when $\log \pi$ is not differentiable would be fruitful \citep{durmus2018efficient}. A  natural extension is to combine SULA with preconditioning. SULA addresses gradient-magnitude robustness, while preconditioned schemes \citep{hird2025quantifying,li2016preconditioned} adapt to the local target geometry; the two mechanisms are complementary and a preconditioned SULA would retain variance balancing while accommodating anisotropic targets. We expect the proof techniques developed here to extend to this setting with appropriate modifications. Finally, like all unadjusted algorithms, SULA performs inexact sampling. A Metropolis correction often employed in each step \citep{Roberts1996} would be required if interest is in directly sampling from $\pi$. The resulting Metropolis-Adjusted SAturated Langevin Algorithm (MASALA) is likely to yield similar robustness benefits as SULA.  
\end{paragraph}

\section*{Acknowledgments and statement on AI use}
Dootika Vats was supported by ANRF under Grant ARGM/2026/0082. The authors utilized
GPT-6 Astra to identify key references across literature and to proofread the final versions of the
paper. 

\bibliographystyle{plainnat}
\bibliography{ref_sula}

\appendix
\section*{Appendix}

\renewcommand{\thelemma}{\thesection.\arabic{lemma}}
\renewcommand{\thetheorem}{\thesection.\arabic{theorem}}
\renewcommand{\theproposition}{\thesection.\arabic{proposition}}
\renewcommand{\theequation}{\thesection.\arabic{equation}}
\renewcommand{\thedefinition}{\thesection.\arabic{definition}}

\section{Notations}\label{sec:D-notations}
We here introduce some notation which will be used throughout this work. Let $\mathbb{R}^d$ denote the $d$-dimensional Euclidean space. For $u, v \in \mathbb{R}^d$, we define the scalar product
$\langle u, v \rangle = \sum_{i=1}^d u_i v_i.$  Let $\| \cdot \|$ denote the Euclidean norm when applied to vectors and $\| \cdot \|_F$ denote the Frobenius norm when applied to matrices.  For any function $f : \mathbb{R}^d \rightarrow \mathbb{R}$ and for some $x = (x_1, x_2, \ldots, x_d) \in \mathbb{R}^d$, let $\partial_i f(x)$ denote the partial derivative with respect to $x_i$. In general, we use the notation
$$
\partial_{i_1\cdots i_k} f(x)
:= \frac{\partial^k f(x)}
{\partial x_{i_1}\cdots \partial x_{i_k}}, \qquad i_1,\dots,i_k \in \{1,\dots,d\}.
$$
Furthermore, let $f^{(k)}(x) := \partial_{i_1\cdots i_k} f(x) \in \mathbb{R}$ denote any $k$-th order derivative of function $f(x)$ with respect to the components of $x \in \mathbb{R}^d$. The space $C^l(\mathbb{R}^d)$ denotes functions from $\mathbb{R}^d$ to $\mathbb{R}$ that are $l$-times continuously differentiable. A function $f$ belongs to $C_P^{l,m}(\mathbb{R}^d)$ if all its partial derivatives up to order $l$ are polynomially bounded, meaning there exists a constant $C > 0$ such that
$
|f^{(k)}(x)| \leq C(1 + \|x\|^m), \text{ for all } x \in \mathbb{R}^d \text{ and } k \leq l.
$
Let $\mathbb{N}$ denote the set of natural numbers. We define the space
$$
C_P^l(\mathbb{R}^d) = \bigcup_{m \in \mathbb{N}} C_P^{l,m}(\mathbb{R}^d),
$$
which contains all $l$-times differentiable functions with polynomial growth. Similarly, we define
$$
C_P^\infty(\mathbb{R}^d) = \bigcap_{l \in \mathbb{N}} C_P^l(\mathbb{R}^d),
$$
which consists of all smooth functions with polynomial growth at infinity. We write $\mathbb{E}_x$ to denote expectation conditioned on the process starting at $x$. Let $A_1, A_2 \in \mathbb{R}^d$. Then $A_1\cdot A_2$ is the dot product.
For some $A \in \mathbb{R}^{d \times d}$, $A^\top$ denotes the transpose of $A$. The indicator function for a set $S$ is denoted by $\mathbf{1}_S$. We denote
the gradient vector by $\nabla f(x)$,
the divergence of the gradient (the Laplacian) by
$$
\nabla \cdot \nabla f(x) = \sum_{i=1}^d \frac{\partial^2 f}{\partial x_i^2}(x),
$$
and the Hessian matrix by
$$
\nabla^2 f(x) = \left( \frac{\partial^2 f}{\partial x_i \partial x_j} (x) \right)_{1 \le i,j \le d}.
$$
More generally, for $k\ge 3$, $\nabla^k f(x)$ is the  $k$-dimensional array with entries
$$
\nabla^k f(x) = \left( \partial_{{i_1}\cdots {i_k}} f(x) \right)_{i_1,\dots,i_k \in \{1,\dots,d\}},
$$
and the corresponding Frobenius norm takes the following form
$$
\left\|\nabla^k f(x)\right\|_F :=
\sqrt{\sum_{i_1,\dots,i_k=1}^d
	\left(\partial_{{i_1}\cdots {i_k}} f(x)\right)^2}.
$$

\section{Weak Convergence Analysis}\label{app:weak-convergence}

By Assumption 1(b) with $y = 0$ and Young's inequality, there exist constants $\widetilde\alpha>0$ and $\widetilde \beta \geq 0$ such that
\[
\langle \nabla \log\pi(x), x \rangle
\leq \widetilde \beta - \widetilde \alpha |x|^2,
\qquad x\in\mathbb R^d.
\]

Moreover, Assumption~\ref{assum:PG} implies the existence of a constant $C_g>0$ such that for $m \ge 1$
\begin{equation*}
	\|\nabla\log \pi(x)\|
	\le C_g\bigl(1+\|x\|^{m}\bigr),
	\qquad \text{for all } x\in\mathbb{R}^d.
\end{equation*}
Now, the following theorem from \cite{skew-symmetric} ensures that expectations of test functions are well-defined and smooth for the Langevin SDE with superlinear drift.
\begin{theorem}[\cite{skew-symmetric}]\label{thm:DiffusionSDE}
	Let $T>0$ be fixed. Under Assumptions~\ref{assum:LC} and \ref{assum:HP},
	consider the Langevin SDE \eqref{eq:langevin} with $Y_0 = x,$ and $ t\in[0,T]$. For a test function $f\in C_P^l(\mathbb{R}^d)$, we define
	$v_t(x) := \mathbb{E}_x[f(Y_t)]$, 
	which is well defined for all $x\in\mathbb{R}^d$ and $t\in[0,T]$, that is
	$\mathbb{E}_x\bigl[\,|f(Y_t)|\,\bigr] < \infty .$
	Moreover, $v_t(\cdot)$ belongs to $C_P^l(\mathbb{R}^d)$ for each $t\in[0,T]$ and there exist constants $C>0$ and an integer $m\ge 1$, independent of $t$ and $x \in \mathbb{R}^d$, such that
	$$
	|v_t(x)| \le C\,(1+\|x\|^m),
	\qquad \text{for all }x\in\mathbb{R}^d.
	$$
	Also, for all $k=1,\dots,l$, we have
	$$
	\left\|\nabla^k v_t(x)\right\|_F \le C\,(1+\|x\|^m).
	$$
\end{theorem}
We now prove the main weak convergence theorem
\subsection{Detailed Proof of Theorem~\ref{thm:weak-true}}
\begin{proof}
	Fix $T>0$,
	and a test function $f \in C_P^4(\mathbb{R}^d)$. Define 
	$$
	v(t,x) = \mathbb{E}_x[f(Y_t)], \qquad Y_0 =x,
	$$
	where $Y_t$ solves the Langevin SDE in (\ref{eq:langevin}). By Theorem \ref{thm:DiffusionSDE}, $v$ is smooth with polynomial growth. For an integer $k$ such that $k h \le T$, we define
	$$
	\varphi_n(x) = v((k-n)h, x), \qquad n=0,1,\dots,k.
	$$
	Note that
	$$
	\varphi_k(x) = v(0,x) = f(x), \qquad \varphi_0(x) = v(k h, x) = \mathbb{E}_x[f(Y_{kh})].
	$$
	We are interested in bounding the weak error
	$$
	\mathbb{E}_x[f(X_k)] - \mathbb{E}_x[f(Y_{k h})] = \mathbb{E}_x[\varphi_k(X_k)] - \varphi_0(x).
	$$
	By the telescoping identity, we write
	\begin{align}\label{eq:weak-telescope}
		\mathbb{E}_x[f(X_k)] - \mathbb{E}_x[f(Y_{k h})] 
		&= \sum_{n=0}^{k-1} \left( \mathbb{E}_x[\varphi_{n+1}(X_{n+1})] - \mathbb{E}_x[\varphi_n(X_n)] \right) \nonumber \\
		&= \sum_{n=0}^{k-1} \mathbb{E}_x[\varphi_{n+1}(X_{n+1}) - \varphi_n(X_n)].
	\end{align}
	It is sufficient to bound each $\mathbb{E}_x[\varphi_{n+1}(X_{n+1}) - \varphi_n(X_n)]$ for each $n$. We can write  the difference as
	$$
	\varphi_{n+1}(X_{n+1}) - \varphi_n(X_n) = \big(\varphi_{n+1}(X_{n+1}) - \varphi_{n+1}(X_n)\big) + \big(\varphi_{n+1}(X_n) - \varphi_n(X_n)\big),
	$$
	and define
	$$
	I_n = \mathbb{E}_x[\varphi_{n+1}(X_{n+1}) - \varphi_{n+1}(X_n)], \qquad
	J_n = \mathbb{E}_x[\varphi_{n+1}(X_n) - \varphi_n(X_n)].
	$$
	Hence
	\begin{equation}\label{eq:phiweak}
		\mathbb{E}_x[\varphi_{n+1}(X_{n+1}) - \varphi_n(X_n)] = I_n + J_n.
	\end{equation}
	Conditioning on $X_n$ and applying Proposition \ref{lmm:Weakexp} with test function $\varphi_{n+1} \in C_P^4$, we get
	$$
	\mathbb{E}_x[\varphi_{n+1}(X_{n+1}) \mid X_n] = \varphi_{n+1}(X_n) + h\, \mathcal{L} \varphi_{n+1}(X_n) + M_1(X_n, h),
	$$
	with $|M_1(x,h)| \le R(x) h^2$, $R \in C_P^\infty$.  
	Hence, we get
	$$
	I_n = \mathbb{E}_x[h \, \mathcal{L} \varphi_{n+1}(X_n) + M_1(X_n, h)].
	$$
	Also, by definition of $\varphi_n$, we can write
	$$
	\varphi_{n+1}(x) - \varphi_n(x) = v((k-n-1)h, x) - v((k-n)h, x),
	$$
	which represents the change of $v$ over a small time step $h$.  
	Since $v$ is smooth and satisfies the backward Kolmogorov equation \cite[p. 46]{Pavliotis2014},
	$$
	\partial_t v(t,x) = \mathcal{L} v(t,x).
	$$
	Performing a first-order Taylor expansion of $v$ in time around $t = (k-n-1)h$, 
	$$
	v((k-n-1)h, x) = v((k-n)h, x) - h\, \partial_t v((k-n-1)h, x) + r(x, h).
	$$
	Replacing the time derivative by $\mathcal{L} v$, 
	we get
	\begin{equation}\label{eq:jn}
		\varphi_{n+1}(x) - \varphi_n(x) = - h \, \mathcal{L} \varphi_{n+1}(x) + r(x, h),    
	\end{equation}
	where the remainder satisfies a polynomial growth bound due to Theorem \ref{thm:DiffusionSDE} 
	$$
	|r(x, h)| \le C_1 \bigl(1 + \|x\|^{m_1}\bigr) h^2,
	$$
	for some constants $C_1, m_1 > 0$.  
	Intuitively, this expresses that the change of $v$ over one backward time step is dominated by the generator $\mathcal{L}$ with a small error of order $h^2$.
	Taking expectation with respect to $X_n$ in (\ref{eq:jn}), we get
	$$
	J_n = \mathbb{E}_x[ - h \, \mathcal{L} \varphi_{n+1}(X_n) + r(X_n, h)].
	$$
	Adding $I_n + J_n$ gives
	\begin{align*}
		I_n + J_n &= \mathbb{E}_x[h \, \mathcal{L} \varphi_{n+1}(X_n)] + \mathbb{E}_x[M_1(X_n, h)] + \mathbb{E}_x[-h \, \mathcal{L} \varphi_{n+1}(X_n)] + \mathbb{E}_x[r(X_n, h)]\\
		&= \mathbb{E}_x[M_1(X_n, h)] + \mathbb{E}_x[r(X_n, h)].
	\end{align*}
	It is clear that both $M_1(X_n, h)$ and $r(X_n, h)$ are polynomially bounded . Hence by Proposition \ref{prop:mombound}, along with (\ref{eq:phiweak}) and (\ref{eq:weak-telescope}) for some $\mathcal{V}(x)  \in C_P^\infty(\mathbb{R}^d),$ 
	we get
	$$
	|\mathbb{E}_x[f(X_k)] - \mathbb{E}_x[f(Y_{k h})] | = \Bigg|\sum_{n=0}^{k-1} (I_n + J_n) \Bigg| \le \mathcal{V}(x) h.
	$$
	This completes the proof.
\end{proof} 

\section{Ergodicity of SULA}
\subsection{Irreducibility, Aperiodicity, and Minorization}
\begin{proposition}[Irreducibility, Aperiodicity, and Minorization]\label{prop:irreducibility}
	The Markov chain generated by the SULA scheme is $\lambda$-irreducible with respect to the Lebesgue measure $\lambda$ on $\mathbb{R}^d$, aperiodic, and also all compact sets are small.
\end{proposition}
\begin{proof}
	Let $P_{h}(x,\cdot)$ be the transition kernel associated with the SULA scheme defined by 
	$$
	P_{h}(x, S) := \mathbb{P}(X_{n+1} \in S \;|\; X_n =x),
	$$
	for any Borel measurable set $S$ of $\mathbb{R}^d$. The distribution of $X_{n+1}$ given $X_n=x$ is Gaussian with mean
	$$
	\mathbb{E}[X_{t+1} \, | \, X_t =x] = x + h\,\drift(x),
	$$
	and covariance matrix $\sigma(x) =$ diag$(\sigma_1(x), \sigma_2(x),\ldots, \sigma_d(x) ) = 2 h\, \diff^2(x).$ 
	Due to the choice of the tuning parameter $c > 1/2$,
	each diagonal entry of $\diff^2(x)$ is strictly positive uniformly for all $x$, and therefore $\sigma(x)$ is a positive definite matrix. As a result, $\sigma(x)$ is invertible and its determinant $\det(\sigma(x)) > 0$. The transition density of the chain takes the following form.
	\begin{equation}\label{eq:SULA-density}
		p_{h}(x,y) = \frac{1}{\sqrt{(2\pi)^d \det(\sigma(x))}} \exp\left( -\frac{1}{2} (y - x - h\,\drift(x))^\top \sigma(x)^{-1} (y - x - h\,\drift(x)) \right),
	\end{equation}
	which is strictly positive for every $x,y \in \mathbb{R}^d$ and finite due to invertibility of $\sigma(x)$. For every $x \in \mathbb{R}^d$ and for any measurable subset $ S \subseteq \mathbb{R}^d $ with positive Lebesgue measure $\lambda$, there exists $\epsilon > 0$ such that
	$P_{h}(x, S) \geq \epsilon > 0.$
	Hence, the $\lambda$-irreducibility and aperiodicity for the scheme holds.

	Let $Q$ be a compact set and fix a compact ball $B_0$ of positive Lebesgue measure. The transition density is continuous and strictly positive, so
	\[
	m_{Q,B_0}:=\inf_{x\in Q,\,y\in B_0}p_h(x,y)>0.
	\]
	Thus, for every Borel set $D$ and every $x\in Q$,
	\[
	P_h(x,D)\geq m_{Q,B_0}\operatorname{Leb}(D\cap B_0)=\epsilon_{Q,B_0}\nu_{B_0}(D),
	\]
	where $\nu_{B_0}$ is the normalized Lebesgue measure on $B_0$ and $\epsilon_{Q,B_0}=m_{Q,B_0}\operatorname{Leb}(B_0)$. Hence every compact set is small. Taking $Q=B_0$ also establishes aperiodicity, since $\nu_{B_0}(B_0)=1$.
\end{proof}

\subsection{Drift Condition}
\begin{proposition}[Drift condition for SULA]\label{prop:RoTULA-drift}
	Under Assumption~\ref{assum1}, there exist constants $a,M,b>0$
	such that, for every $h\in(0,1)$ and $x\in\mathbb{R}^d$,
	\begin{equation}\label{eq:RoTULA-drift-alt}
		P_h W_a(x)
		\le e^{-a^2h}W_a(x)
		+ hb\,\mathbf{1}_{\{\|x\|\le M\}},
	\end{equation}
	where $P_h$ is the Markov transition kernel of SULA and
	$W_a(x)=\exp\bigl(a\sqrt{1+\|x\|^2}\bigr)$.
\end{proposition}
\begin{proof}
	Consider the update $Y := X_{t+1} = X_t + h\,\drift(X_t) + \sqrt{2h} \, \diff(X_t) Z$, where $Z \sim \mathcal{N}(0, I_d)$. Throughout the proof, we condition on $X_t=x$. Then $P_h(x,\cdot)$
	is the law of $Y$, with mean $x+h\,\drift(x)$ and covariance
	matrix $2h\,\diff^2(x)$.
	The Lyapunov function is $W_a(x) = \exp\bigl( a \varphi(x) \bigr)$, where $\varphi(x) = \left(1 + \|x\|^2\right)^{1/2}$ and $a > 0$. 
	Note that $\varphi(\cdot)$ is $1$-Lipschitz since $\|\nabla \varphi(x)\| \le 1$ for all $x$. 
	For $z\in\mathbb{R}^d$, define
	\[
	u(z)=\varphi\bigl(x+h\,\drift(x)+\sqrt{2h}\,\diff(x)z\bigr).
	\]
	Since $\varphi$ is $1$-Lipschitz, for all $z_1,z_2\in\mathbb{R}^d$,
	\[
	|u(z_1)-u(z_2)|
	\le \sqrt{2h}\,\|\diff(x)(z_1-z_2)\|
	\le \sqrt{2h}\,\|\diff(x)\|_F\,\|z_1-z_2\|.
	\]
	Hence, due to $\|\diff(x)\|_F \le \sqrt{d}$, the Lipschitz constant $L$ of $u$ satisfies
	\begin{equation}\label{eq:lip-u}
		L \le \sqrt{2h} \, \left\|\diff(x)\right\|_F \le \sqrt{2hd} .
	\end{equation}
	The Gaussian Herbst argument states that for any $L$-Lipschitz function $u$ of standard Gaussian $Z$ and $\lambda \in \mathbb{R}$, we have \cite[Proposition 5.5.1]{Bakry2013}
	$$
	\mathbb{E}\left[ e^{\lambda (u(Z) - \mathbb{E} u(Z))} \right] \le \exp\left( \frac{\lambda^2 L^2}{2} \right).
	$$
	Taking $\lambda = a > 0$ gives
	$$
	\mathbb{E}\left[ e^{a u(Z)} \right] \le \exp\left( a \mathbb{E}[u(Z)] + \frac{a^2 L^2}{2} \right).
	$$
	From \eqref{eq:lip-u}, $L^2/2 \le hd$, so $\frac{a^2 L^2}{2} \le a^2 hd$. Thus
	\begin{equation}\label{eq:herbst}
		P_{h} W_a(x) = \mathbb{E}\bigl[ e^{a \varphi(Y)} \bigr] \le e^{a^2 hd} \exp\bigl( a \mathbb{E}[\varphi(Y)] \bigr).
	\end{equation} 
	For $Y$, we have
	$$
	\mathbb{E}\|Y\|^2 = \|x + h\,\drift(x)\|^2 + 2h \operatorname{tr}\left(\diff^2(x)\right) \le \|x + h\,\drift(x)\|^2 + 2h d.
	$$
	The function $\varphi(y) = (1 + \|y\|^2)^{1/2}$ satisfies $\mathbb{E}[\varphi(V)] \le (1 + \mathbb{E}\|V\|^2)^{1/2}$ by Jensen’s inequality for any random vector $V$. 
	Therefore \eqref{eq:herbst} gives
	\begin{equation}\label{eq:working}
		P_{h} W_a(x) \le e^{a^2 hd} \exp\left\{ a \bigl( 1 + \|x + h\,\drift(x)\|^2 + 2h d \bigr)^{1/2} \right\}.
	\end{equation}
	From Lemma~\ref{lem:coercivity}, we have constants $M_1$ and $\kappa'$ such that multiplying \eqref{eq:meanmatcheddrift} by
	$2h\|x\|$ gives, for $\|x\|\ge M_1$,
	$$
	\|x + h\,\drift(x)\|^2 \le \|x\|^2 - 2h \kappa' \|x\|.
	$$
	Adding the variance bound gives
	\begin{equation}\label{eq:mean-square}
		\|x + h\,\drift(x)\|^2 + 2h d \le \|x\|^2 - 2h \kappa' \|x\| + 2h d, \quad \|x\| \ge M_1.
	\end{equation}
	Choose $M \ge \max(M_1, 2d/\kappa')$. For $\|x\| \ge M$, we get
	$$
	\|x + h\,\drift(x)\|^2 + 2h d \le \|x\|^2 - 2h \kappa' \|x\| + \kappa' \|x\| h = \|x\|^2 - \kappa' \|x\| h.
	$$
	Thus
	\begin{equation}\label{eq:sqrt-outside}
		\bigl(1 + \|x + h\,\drift(x)\|^2 + 2h d\bigr)^{1/2} \le \bigl(1 + \|x\|^2 - \kappa' \|x\| h\bigr)^{1/2}, \quad \|x\| \ge M.
	\end{equation}
	We write the right-hand side of \eqref{eq:sqrt-outside} as $\varphi(x)(1-t)^{1/2}$ where $t = \kappa' \|x\| h/(1 + \|x\|^2)$.  
	By \eqref{eq:sqrt-outside}, $1+\|x\|^2-\kappa'\|x\|h\ge 0$,
	so $0\le t\le 1$ for every $h\in(0,1)$. Hence
	$(1-t)^{1/2}\le 1-t/2$, which gives
	\[
	\bigl(1+\|x+h\,\drift(x)\|^2+2hd\bigr)^{1/2}
	\le \varphi(x)\left(1-\frac{t}{2}\right)
	= \varphi(x)-\frac{\kappa'\|x\|h}{2\varphi(x)}.
	\]
	Since $s \mapsto s/\varphi(s)$ is non-decreasing for $s \ge 0$, we have $\|x\|/\varphi(x) \ge M/\varphi(M)$ for $\|x\| \ge M$. Define
	$$
	a := \frac{\kappa' M}{4 \varphi(M) d}.
	$$
	Then $\kappa' \|x\| h/(2 \varphi(x)) \ge 2 a h d$, so
	\begin{equation}\label{eq:contraction-outside}
		\bigl(1 + \|x + h\,\drift(x)\|^2 + 2h d \bigr)^{1/2} \le \varphi(x) - 2 a h d, \quad \|x\| \ge M.
	\end{equation}
	Substituting \eqref{eq:contraction-outside} into \eqref{eq:working} with $a = \frac{\kappa' M}{4 \varphi(M) d}$
	$$
	P_{h} W_a(x) \le e^{a^2 hd} \exp\bigl( a (\varphi(x) - 2 a h d) \bigr) = e^{a^2 h d- 2 a^2 h d} W_a(x) = e^{-a^2 hd} W_a(x) \le e^{-a^2 h} W_a(x).
	$$
	Thus
	\begin{equation}\label{eq:outside-ball}
		P_{h} W_a(x) \le e^{-a^2 h} W_a(x), \quad \|x\| \ge M.
	\end{equation}
	Now consider $\|x\|\le M$. Let
	$G_M=\sup_{\|y\|\le M}\|\nabla\log\pi(y)\|$, which is finite
	by continuity. Since $\|\drift(x)\|\le G_M$ and $h<1$,
	\[
	\|x+h\,\drift(x)\|^2+2hd
	\le \|x\|^2+2hC_1,
	\qquad C_1=MG_M+\tfrac12 G_M^2+d.
	\]
	As $\varphi(x)=\sqrt{1+\|x\|^2}\ge 1$, we obtain
	\[
	\bigl(1+\|x+h\,\drift(x)\|^2+2hd\bigr)^{1/2}
	\le \bigl(\varphi(x)^2+2hC_1\bigr)^{1/2}
	\le \varphi(x)+hC_1.
	\]
	Substituting into \eqref{eq:working} gives
	\[
	P_hW_a(x)\le e^{C_2h}W_a(x),
	\qquad C_2=a^2d+aC_1.
	\]
	Since $1-e^{-u}\le u$ for $u\ge0$ and $h<1$,
	\[
	e^{C_2h}-e^{-a^2h}
	=e^{C_2h}\bigl(1-e^{-(C_2+a^2)h}\bigr)
	\le h e^{C_2}(C_2+a^2).
	\]
	Using $W_a(x)\le e^{a\varphi(M)}$ on $\|x\|\le M$, define
	\begin{equation}\label{eq:b-constant}
		b=e^{C_2}(C_2+a^2)e^{a\varphi(M)}.
	\end{equation}
	Then $P_hW_a(x)\le e^{-a^2h}W_a(x)+hb$ for $\|x\|\le M$.
	Combining this with \eqref{eq:outside-ball} proves
	\eqref{eq:RoTULA-drift-alt}.
\end{proof}


\section{Bias of the Invariant Measure}\label{app:bias}
\subsection{Proof of Theorem \ref{thm:bias}}
\begin{proof}
	We prove the result by checking the assumptions in Theorem 3.3 of \cite{Lelievre2018}. We start with the choice of the test function space. We take $S = C_P^\infty(\mathbb{R}^d),$
	the space of smooth functions with polynomial growth. The space $S$ is dense in $L^2(\pi)$.
	Next, from Proposition \ref{lmm:Weakexp}, the one-step operator satisfies
	$$
	P_{h} f
	=
	f + h\, \mathcal{L} f + h^2\, \mathcal{A} f + M_2(x, h).
	$$
	So in the notations of \cite{Lelievre2018}, we identify $\mathcal{A}_1 = \mathcal{L}$ and $\mathcal{A}_2 = \mathcal{A}$. This shows that the scheme is consistent with the generator $\mathcal{L}$. Now we use the invariance of $\pi$. Since $\pi$ is invariant for the diffusion, hence for any $f \in C_P^\infty(\mathbb{R}^d)$,
	$$
	\int_{\mathbb{R}^d} P_t f(x)\,\pi(dx)
	=
	\int_{\mathbb{R}^d} f(x)\,\pi(dx)
	\qquad \text{for all } t\ge 0,
	$$
	where $(P_t)_{t\ge 0}$ is the semigroup of the diffusion. 
	By definition of the generator, we have
	$$
	\mathcal{L}f(x)
	=
	\lim_{t\to 0}\frac{P_t f(x)-f(x)}{t}.
	$$
	Since $f \in C_P^\infty(\mathbb{R}^d)$, we also have $\mathcal{L}f \in C_P^\infty(\mathbb{R}^d)$, and under Assumption \ref{assum:bias-result}, the measure $\pi$ has moments of all orders. Thus
	$$
	\left. 0 = \frac{d}{dt}\int_{\mathbb{R}^d} P_t f(x)\,\pi(dx)\right|_{t=0}
	=
	\int_{\mathbb{R}^d} \left.\frac{d}{dt}P_t f(x)\right|_{t=0}\,\pi(dx)
	=
	\int_{\mathbb{R}^d} \mathcal{L}f(x)\,\pi(dx).
	$$
	Since the left-hand side is the derivative of a constant function, it is equal to $0$. Hence
	$$
	\int_{\mathbb{R}^d} \mathcal{L}f(x)\,\pi(dx)=0.
	$$
	We now define the space
	$
	\mathcal{S}_0 = \left\{ f \in C_P^\infty(\mathbb{R}^d) : \int f\, d\pi = 0 \right\}.
	$
	From the above identity, it is clear that $\mathcal{S}_0$ is invariant under $\mathcal{L}$. Next, let $\mathcal{L}^*$ be the adjoint of $\mathcal{L}$ in $L^2(\pi)$. Since the diffusion is reversible, we have
	$
	\mathcal{L}^* = \mathcal{L}.
	$
	Following \citet[Appendix~E]{skew-symmetric}, let $g\in\mathcal S_0$.
	By Assumption~\ref{assum:LC} with $y=0$ and uniform ellipticity,
	\citet[Theorem~1]{pardoux2001poisson} gives a unique centred solution
	$u$ of polynomial growth to $\mathcal Lu=g$. Since
	$\Delta u=g-\nabla\log\pi\cdot\nabla u$, successive interior Sobolev
	estimates on balls centred at $x$, followed by Sobolev embedding,
	give polynomial bounds in $\|x\|$ for every derivative of $u$:
	the interior estimates for the Laplacian are invariant under
	translation, and all derivatives of $g$ and $\nabla\log\pi$ have
	polynomial growth. Thus $u\in\mathcal S_0$, and $\mathcal L$ is
	invertible on $\mathcal S_0$.
	Now we look at the operator $\mathcal{A}$. Under Assumption \ref{assum:bias-result}, we have $\log \pi \in C_P^4(\mathbb{R}^d)$. This ensures that derivatives of $\pi$ divided by $\pi$ are polynomially bounded. Using the expression of $\mathcal{A}$ and integrating by parts, for $\mathbb{I} : \mathbb{R}^d\rightarrow \mathbb{R}$ which is equal to 1 everywhere, we get the following
	\begin{align*}
		\mathcal{A}^\ast \mathbb{I}(x) 
		&= \frac{c}{2} \sum_{i=1}^d \frac{\partial_i \big( (\nabla_i \log \pi(x))^3 \, \pi(x) \big)}{\pi(x)}
		+ \frac{1}{2} \sum_{i,j=1}^d \frac{\partial_{ij} \big( (\nabla_i \log \pi(x))(\nabla_j \log \pi(x)) \, \pi(x) \big)}{\pi(x)} \\
		&\quad - \frac{1}{2} \sum_{i=1}^d \frac{\partial_{ii} \big( (\nabla_i \log \pi(x))^2 \, \pi(x) \big)}{\pi(x)}
		- \sum_{i,j=1}^d \frac{\partial_{ijj} \big( (\nabla_i \log \pi(x)) \, \pi(x) \big)}{\pi(x)} \\
		&\quad + \frac{1}{2} \sum_{i,k=1}^d \frac{\partial_{iikk} \pi(x)}{\pi(x)}.
	\end{align*}
	From the assumptions on the drift and $\log \pi$, all terms above are smooth and have polynomial growth. Hence
	$
	\mathcal{A}^\ast \mathbb{I} \in C_P^\infty(\mathbb{R}^d).
	$
	Also, since $\mathcal{A}$ is a differential operator, we have $\mathcal{A}\mathbb{I} = 0$. Therefore,
	$$
	\int \mathcal{A}^\ast \mathbb{I}(x)\, \pi(x)\, dx
	=
	\int \mathcal{A}\mathbb{I}(x)\, \pi(x)\, dx
	= 0.
	$$
	Thus, $\mathcal{A}^\ast\mathbb{I}\in\mathcal{S}_0$. Since $\mathcal{L}$ is invertible on $\mathcal{S}_0$, there exists a unique $\phi\in\mathcal{S}_0$ such that
	$
	\mathcal{L}\phi=-\mathcal{A}^\ast\mathbb{I}.
	$
	Hence, the correction function in Theorem~\ref{thm:bias} is
	$
	\phi=-\mathcal{L}^{-1}(\mathcal{A}^\ast\mathbb{I}).
	$
	Finally, from Theorem \ref{thm:geom-ergodicity}, the scheme admits a unique invariant measure $\pi_{h}$ and this measure has finite moments. This ensures that all required integrals are well-defined. Further, Proposition~\ref{prop:RoTULA-drift} provides a Lyapunov drift condition with constants independent of $h\in(0,1)$. Integrating \eqref{eq:RoTULA-drift-alt} with respect to the invariant measure $\pi_h$, for which the Lyapunov function $W_a$ is integrable by the drift condition, yields
	\[
	\int W_a\,d\pi_h\leq \frac{hb}{1-\exp(-a^2h)}.
	\]
	Since the right-hand side is uniformly bounded over $h\in(0,1)$ and $W_a$ dominates every polynomial in $\|x\|$, it follows that
	\[
	\sup_{h\in(0,1)}\int\|x\|^p\,d\pi_h<\infty
	\]
	for every $p\geq1$. Thus, the invariant measures satisfy the required uniform moment bounds. Therefore, applying Theorem 3.3 of \cite{Lelievre2018} completes the proof.
\end{proof}

\section{Results for SG-SULA}\label{app:sg}
\subsection{Proof of Proposition \ref{sgprop1}}

\begin{proof}
	We fix $i \in \{1,\dots,d\}$ and condition on $X_t$. Recall that
	$$
	\widehat{\mathcal{S}}_h^{(i)}(X_t, 1)
	=
	\frac{\widehat{\nabla}_i \log \pi(X_t)}{
		\sqrt{1 + c\,h \big(\widehat{\nabla}_i \log \pi(X_t)\big)^2}},
	\qquad
	\drift^{(i)}(X_t)
	=
	\frac{\nabla_i \log \pi(X_t)}{
		\sqrt{1 + c\,h \big(\nabla_i \log \pi(X_t)\big)^2}}.
	$$
	Using
	$$
	\widehat{\nabla}_i \log \pi(X_t)
	=
	\nabla_i \log \pi(X_t) + \eta_t^{(i)},
	$$
	we obtain
	$$
	\mathbb{E}\left[\widehat{\mathcal{S}}_h^{(i)}(X_t, 1)\mid X_t\right]
	=
	\mathbb{E}\left[
	\frac{\nabla_i \log \pi(X_t) + \eta_t^{(i)}}{
		\sqrt{1 + c\,h \big(\nabla_i \log \pi(X_t) + \eta_t^{(i)}\big)^2}}
	\;\middle|\; X_t
	\right].
	$$
	Since $\eta_t^{(i)} \stackrel{d}{=} -\eta_t^{(i)}$, we may write
	$$
	\mathbb{E}\left[\widehat{\mathcal{S}}_h^{(i)}(X_t, 1)\mid X_t\right]
	=
	\mathbb{E}\left[
	\frac{1}{2}\left(
	\frac{\nabla_i \log \pi(X_t) + \eta_t^{(i)}}{
		\sqrt{1 + c\,h \big(\nabla_i \log \pi(X_t) + \eta_t^{(i)}\big)^2}}
	+
	\frac{\nabla_i \log \pi(X_t) - \eta_t^{(i)}}{
		\sqrt{1 + c\,h \big(\nabla_i \log \pi(X_t) - \eta_t^{(i)}\big)^2}}
	\right)
	\;\middle|\; X_t
	\right].
	$$
	If $\nabla_i \log \pi(X_t) = 0$, then by Lemma~\ref{lmm:propSG2}, the quantity inside the expectation is equal to $0$. Hence
	$$
	\mathbb{E}\left[\widehat{\mathcal{S}}_h^{(i)}(X_t, 1)\;\middle|\; X_t\right] = 0 = \drift^{(i)}(X_t).
	$$
	Let $\nabla_i \log \pi(X_t) > 0$. We define
	$$
	\mathcal{H}\left(\eta_t^{(i)}\right)
	=
	\frac{1}{2}\left(
	\frac{\nabla_i \log \pi(X_t) + \eta_t^{(i)}}{
		\sqrt{1 + c\,h \big(\nabla_i \log \pi(X_t) + \eta_t^{(i)}\big)^2}}
	+
	\frac{\nabla_i \log \pi(X_t) - \eta_t^{(i)}}{
		\sqrt{1 + c\,h \big(\nabla_i \log \pi(X_t) - \eta_t^{(i)}\big)^2}}
	\right).
	$$
	By Lemma~\ref{lmm:propSG2}, $\mathcal{H}\left(\eta_t^{(i)}\right) > 0$. Hence
	\begin{equation}\label{eq:E-geq}
		\mathbb{E}\left[\widehat{\mathcal{S}}_h^{(i)}(X_t, 1)\;\middle|\; X_t\right] > 0.
	\end{equation}
	On the other hand, by Lemma~\ref{lmm:propSG1}, we have
	$$
	\mathcal{H}\left(\eta_t^{(i)}\right)
	\le
	\frac{\nabla_i \log \pi(X_t)}{
		\sqrt{1 + c\,h \big(\nabla_i \log \pi(X_t)\big)^2}}
	=
	\drift^{(i)}(X_t)
	$$
	with equality if $\eta_t^{(i)}=0$. Taking expectations along with \eqref{eq:E-geq} yields 
	\begin{equation}\label{eq:grad-less}
		0 < \mathbb{E}\left[\widehat{\mathcal{S}}_h^{(i)}(X_t, 1)\;\middle|\; X_t\right]
		\le
		\drift^{(i)}(X_t).
	\end{equation}
	On similar lines, when $\nabla_i \log \pi(X_t) < 0,$ we get
	\begin{equation}\label{eq:grad-gra}
		\drift^{(i)}(X_t) \le \mathbb{E}\left[\widehat{\mathcal{S}}_h^{(i)}(X_t, 1)\;\middle|\; X_t\right] < 0.
	\end{equation}
	Therefore, combining \eqref{eq:grad-less} and \eqref{eq:grad-gra} gives
	$$
	\left|\mathbb{E}\left[\widehat{\mathcal{S}}_h^{(i)}(X_t, 1)\;\middle|\; X_t\right]\right|
	\le
	\left|\drift^{(i)}(X_t)\right|.
	$$
	This completes the proof.
\end{proof}

\subsection{Proof of Proposition \ref{sgprop2}}

\begin{proof}
	Since our interest is to claim the following for fixed $z \in \mathbb{R}$ and for a fixed $\alpha$,
	$$
	\mathbb{E}_{\eta_t^{(i)}}\left[
	\Phi\left(\alpha\sqrt{ch}\,\widehat{\nabla}_i \log \pi(X_t)z\right)
	\right]-
	\Phi\left(\sqrt{ch}\,\nabla_i \log \pi(X_t)z\right) =0.
	$$
	To prove our claim, we assume
	$
	\widehat{\nabla}_i \log \pi(X_t)
	= \nabla_i \log \pi(X_t) + \eta^{(i)}_t$ with $ \eta^{(i)}_t \sim \mathcal{N}(0, \tau^2),
	$
	and therefore
	\begin{align*}
		\alpha \sqrt{ch}\left(\nabla_i \log \pi(X_t) + \eta^{(i)}_t\right) z 
		\sim \mathcal{N}\left(
		\alpha \sqrt{ch}\,z \nabla_i \log \pi(X_t),\;
		c\,h\, \alpha^2 z^2 \tau^2
		\right).
	\end{align*}
	Consequently,
	\begin{align*}
		\mathbb{E}_{\eta^{(i)}_t}\left[
		\Phi\left(\alpha \sqrt{ch}\left(\nabla_i \log \pi(X_t) + \eta^{(i)}_t\right)z\right)
		\right]
		=
		\Phi\left(
		\frac{\alpha
			\sqrt{ch}\,z \nabla_i \log \pi(X_t)
		}{
			\sqrt{1+c\,h\, \alpha^2 z^2 \tau^2}
		}
		\right).
	\end{align*}
	We observe that the target gradient is shrunk by the multiplicative factor less than $1$.
	To remove the shrinkage, we
	choose $\alpha = \alpha_{z,h}^{\tau}$ such that
	$$
	\frac{\alpha}{\sqrt{1 + \alpha^2 c\,h\, z^2 \tau^2}} = 1
	\implies 
	\alpha_{z,h}^{\tau}
	=
	\frac{1}{\sqrt{1 - c\,h\, z^2 \tau^2}},
	\quad  |z| < \frac{1}{\tau \sqrt{c h}}.
	$$
	With this choice,
	\begin{align*}
		\mathbb{E}_{\eta_t^{(i)}}\left[
		\Phi\left(\sqrt{ch}\,\alpha_{z,h}^{\tau}\widehat{\nabla}_i \log \pi(X_t)z\right)
		\right]
		=
		\Phi\left(
		\sqrt{ch}\,
		\nabla_i \log \pi(X_t)\,z
		\right),
	\end{align*}
	which proves our claim.
\end{proof}

\section{Supporting Results}

We prove some auxiliary lemmas here that support our previous results.

\begin{lemma}\label{thm:main}
	For all $u\ge 0$ and all $r\ge 1$,
	$$
	(1+u)^r \le 2^{\,r-1}\bigl(1+u^r\bigr).
	$$
	Equality holds iff $r=1$ or $u=1$.
\end{lemma}
\begin{proof}
	The function $f(t) = t^r$ is convex on $[0,\infty)$ for $r\ge 1$. By definition of a convex function, we have 
	$$\left(\frac{1 +u}{2}\right)^r = \left(\frac{1}{2} + \frac{u}{2}\right)^r \le \frac{1}{2} 1^r + \frac{1}{2}u^r = \frac{1^r +u^r}{2},$$
	which gives
	$$(1+u)^r \le 2^{r-1} (1+u^r).$$
	Hence, the result is proved.
\end{proof}

\begin{lemma}\label{prop:ineq-2}
	For every $a\ge 0$,
	$$
	\frac{a}{\sqrt{1+a}\,\bigl(1+\sqrt{1+a}\bigr)} \le \frac{a}{2}.
	$$
\end{lemma}

\begin{proof}
	Let $a\ge0$. Then
	$$
	1+a \ge 1 \quad\Rightarrow\quad \sqrt{1+a} \ge 1,
	$$
	and 
	\begin{align*}
		\sqrt{1+a}\,\bigl(1+\sqrt{1+a}\bigr) \ge 1\cdot 2 = 2.
	\end{align*}
	Since the product on the left is positive, taking reciprocals gives
	$$
	\frac{1}{\sqrt{1+a}\,\bigl(1+\sqrt{1+a}\bigr)} \le \frac{1}{2},
	$$
	as required. 
\end{proof}
\begin{lemma}\label{lm:sqrd_norm}
	From Assumption \ref{assum:PG}, there exist constants $C_g>0$ and $m\ge1$ such that
	$$
	\|\nabla \log \pi(x)\| \le C_g\bigl(1+\|x\|^m\bigr),
	\qquad x\in\mathbb{R}^d.
	$$
	Then, for any $r \ge 1$,
	$$
	\|\nabla \log \pi(x)\|^r 
	\le 
	2^{\,r-1} C_g^r \bigl(1+\|x\|^{mr}\bigr),
	\qquad x\in\mathbb{R}^d.
	$$
\end{lemma}
\begin{proof}
	By raising the assumed bound to the power $r \ge 1$ gives
	$$
	\|\nabla \log \pi(x)\|^r 
	\le 
	C_g^r \bigl(1+\|x\|^m\bigr)^r.
	$$
	Using Lemma \ref{thm:main}, and applying it with $u=\|x\|^m$, we obtain
	$$
	\bigl(1+\|x\|^m\bigr)^r 
	\le 
	2^{\,r-1}\bigl(1+\|x\|^{mr}\bigr).
	$$
	Combining the bounds yields the result.
\end{proof}

In the following lemma, we quantify the discretization error in the SULA coefficients by comparing them with the true Langevin drift and diffusion terms.

\begin{lemma}\label{lmm:coeffbound}
	Under Assumption \ref{assum:PG}, for each $x \in \mathbb{R}^d$, the coefficients of SULA satisfy
	\begin{align*}
		\| \drift(x) - \nabla \log \pi(x) \| 
		&\leq 
		2c\, C_g^3 \left(1 + \|x\|^{3m}\right) h,\\
		\left \| \diff(x)^2 - I_d \right\|_F 
		&\leq 
		C_g^2 \left(1 + \|x\|^{2m}\right) h.
	\end{align*}
\end{lemma}

\begin{proof}
	We recall drift of SULA scheme with $i$th component  
	$$\drift^{(i)}(x) = 
	\frac{\nabla_i \log \pi(x)}{\sqrt{1 + c\,h (\nabla_i\log\pi(x))^2}}.
	$$
	We  consider
	\begin{align*}
		\left|\drift^{(i)}(x)-\nabla_i\log\pi(x)\right| &= |\nabla_i\log\pi(x)| 
		\left|\frac{1}{\sqrt{1 + c\,h (\nabla_i\log\pi(x))^2}} - 1\right|\\
		&= |\nabla_i\log\pi(x)| \left|\frac{1-\sqrt{1 + c h (\nabla_i\log\pi(x))^2}}{\sqrt{1 + c h (\nabla_i\log\pi(x))^2}}\right|\\
		&= |\nabla_i\log\pi(x)|
		\left|\frac{-c\,h (\nabla_i\log\pi(x))^2}
		{\sqrt{1 + c\,h (\nabla_i\log\pi(x))^2}
			\left(1+\sqrt{1 + c\,h (\nabla_i\log\pi(x))^2}\right)}\right|. 
	\end{align*}
	Using the relation from Lemma \ref{prop:ineq-2},
	we get
	\begin{align*}
		\left|\drift^{(i)}(x)-\nabla_i\log\pi(x)\right|
		&\le \frac{c}{2}h |\nabla_i\log\pi(x)|^3.
	\end{align*}
	Square and sum these component-wise bounds to control the Euclidean norm 
	\begin{equation}\label{eq:A_normed-squared}
		\|\drift(x)-\nabla \log\pi(x)\|^2
		= \sum_{i=1}^d \left|\drift^{(i)}(x) - \nabla_i\log\pi(x)\right|^2
		\le \Big(\frac{c}{2}\,h\Big)^2 \sum_{i=1}^d |\nabla_i\log\pi(x)|^6.
	\end{equation}	
	Now using monotonicity of $\ell^p$ norms
	$\|v\|_6 \le \|v\|, \; v \in \mathbb{R}^d$ and using Lemma \ref{lm:sqrd_norm},
	we get
	$$
	\left(\sum_{i=1}^d |\nabla_i \log\pi(x)|^6\right)^{1/2}
	= \left(\|\nabla \log\pi(x)\|_6\right)^3
	\le \|\nabla \log\pi(x)\|^3 \le 4C_g^3 (1+\|x\|^{3m}).
	$$
	Taking square roots in (\ref{eq:A_normed-squared}) gives the desired bound
	$$
	\|\drift(x)-\nabla\log\pi(x)\| \le \frac{c}{2}\,h \left(\sum_{i=1}^d |\nabla_i \log\pi(x)|^6\right)^{1/2}
	\le 2c\,C_g^3 (1+\|x\|^{3m})h.
	$$
	Now, we write
	$$
	\diff^{(i)}(x)^2 
	= 1 - \frac{1}{2}\frac{h (\nabla_i\log\pi(x))^2}
	{1 + c\,h(\nabla_i\log\pi(x))^2},
	$$
	and since the denominator is $\ge 1$ which leads to
	\begin{equation}\label{eq:cov-combound}
		\left|\diff^{(i)}(x)^2 -1 \right|
		\le \frac{1}{2}h (\nabla_i\log\pi(x))^2.
	\end{equation}
	Therefore on a similar line, we write
	$$ \left\|\diff(x)^2 -I_d \right\|_F^2 = \sum_{i=1}^{d} \left|\diff^{(i)}(x)^2 - 1\right|^2 \le \sum_{i=1}^{d} \left(\frac{1}{2}h(\nabla_i\log\pi(x))^2\right)^2 = \frac{1}{4} h^2 \sum_{i=1}^{d}(\nabla_i\log\pi(x))^4.$$
	Using $\|v\|_4 \le \|v\|, \; v \in \mathbb{R}^d$ gives
	$$
	\left\|\diff(x)^2-I_d\right\|_F
	\le \frac{1}{2}h \|\nabla\log\pi(x)\|^2
	\le C_g^2 (1+\|x\|^{2m})h.
	$$
	This completes the proof.
\end{proof}

With the aim of proving Theorem~\ref{thm:weak-true}, which establishes that SULA approximates the underlying SDE in the limit as the step-size $h \to 0$, we now introduce the operators governing the leading-order terms in the corresponding Taylor expansions and present an important remark.

\begin{definition}
	For $f \in C^2(\mathbb{R}^d)  $, the generator $\mathcal{L}$ of the Langevin diffusion is defined by
	\begin{equation*}
		\mathcal{L}f(x) = \nabla \log \pi (x) \cdot\nabla f + \nabla \cdot \nabla f.
	\end{equation*}
	Also, for $f \in C^4(\mathbb{R}^d)$, we introduce an operator $\mathcal{A}: C^4(\mathbb{R}^d) \rightarrow C^0(\mathbb{R}^d)$ as
	\begin{align}\label{eq:RoTULA-A2-operator}
		\mathcal{A} f(x)
		\;=\; &-\frac{c}{2}\sum_{i=1}^d \partial_i f(x)\,\bigl(\nabla_i\log\pi(x)\bigr)^3 +\frac{1}{2}\sum_{i,j=1}^d \partial_{ij} f(x)\,\nabla_i\log\pi(x)\,\nabla_j\log\pi(x) \nonumber \\
		&\;-\;\frac{1}{2}\sum_{i=1}^d \partial_{ii} f(x)\,\left(\nabla_i\log\pi(x) \right)^2
		+\sum_{i, j=1}^d \nabla_i\log\pi(x)\partial_{ijj} f(x) +\frac{1}{2}\sum_{i,k=1}^d \partial_{i i k k} f(x).
	\end{align}
\end{definition}

It is noted that using the Binomial expansion for small $u$  
$$
(1+u)^{-1/2} = 1 - \frac{1}{2}u + \frac{3}{8}u^2 - \frac{5}{16}u^3 + \cdots,
$$
we write $\drift^{(i)}(x)$ by setting $u=c\,h (\nabla_i\log\pi(x))^2$ as
\begin{align}\label{eq:A-R_expansion}
	\drift^{(i)}(x) 
	&= \nabla_i\log\pi(x)\left(
	1 - \frac{1}{2}c\,h (\nabla_i\log\pi(x))^2
	+ O(h^2)
	\right) \nonumber\\
	&=: \nabla_i\log\pi(x) + R^{(i)}_{A}(x,h),
\end{align}
where
\begin{equation}\label{eq:RA}
	R^{(i)}_{A}(x,h)
	= \nabla_i\log\pi(x)
	\left(
	-\frac{1}{2}c\,h (\nabla_i\log\pi(x))^2
	+ O(h^2)
	\right).
\end{equation}
Under Assumption \ref{assum:PG} and with the help of Lemma \ref{lm:sqrd_norm} for all $i$, we get
\begin{equation}\label{eq:R-bound}
	\left|R^{(i)}_{A}(x,h)\right|
	\le 2c\, C_g^3(1+\|x\|^{3m})h
	+ O\bigl(c^2C_g^5(1+\|x\|^{5m})h^2\bigr).
\end{equation}

Based on the setup above, we now state and prove the following result that describes the local discretization error of the SULA scheme.

\begin{proposition}[One-step weak expansion for SULA] \label{lmm:Weakexp}
	Under Assumptions \ref{assum:LC}, \ref{assum:HP}, and \ref{assum:PG}, we have the following:
	\begin{enumerate}
		\item[(i)] For any function $f\in C_P^4(\mathbb R^d)$, there exist $M_1 : \mathbb{R}^d \times (0,1) \rightarrow \mathbb{
			R}$ and $R_1(x) \in C_P^\infty(\mathbb R^d)$ such that for every $h\in(0, 1)$ and for every $x\in\mathbb R^d$
		\begin{equation}\label{eq:one-step-expansion}
			\mathbb{E}_x\big[f(X_{1})\big]
			= f(x) + h \cdot \mathcal{L}f(x) + M_1(x,h),
		\end{equation}
		and the remainder obeys the bound
		$
		|M_1(x,h)| \le R_1(x)\,h^2.$
		\item[(ii)] For any $f\in C_P^{6}(\mathbb R^d)$
		there exists $M_2 : \mathbb{R}^d \times (0,1) \rightarrow \mathbb{
			R}$  and $R_2(x)\in C_P^\infty(\mathbb R^d)$ such that for every $h\in(0, 1)$ and for every $x\in\mathbb R^d$
		\begin{equation}\label{eq:RoTULA-secondorder-expansion}
			\mathbb{E}_x[f(X_1)]
			= f(x) + h\cdot\mathcal L f(x) + h^2\cdot\mathcal{A} f(x)
			+ M_2(x,h),
		\end{equation}
		where the remainder satisfies $
		|M_2(x,h)| \le R_2(x) h^3.$
	\end{enumerate}
\end{proposition}

\begin{proof}
	Consider the first step of SULA with $X_0 = x$ and let
	\begin{equation}\label{eq:firststep}
		\Delta X = X_1 - X_0 = h\,\drift(x)
		\;+\; \sqrt{2h}\,\diff(x)\, Z_1.
	\end{equation}
	We first prove part (ii), so let $f\in C_P^6(\mathbb R^d)$.
	For $\Delta X = (\Delta X^{(1)}, \Delta X^{(2)}, \ldots, \Delta X^{(d)}) \in \mathbb{R}^d$, a Taylor expansion of $f$ around $x$ with integral remainder yields
	\begin{align*}
		\mathbb{E}_x&\left[f(x + \Delta X)\right]
		= f(x)
		+ \mathbb{E}_x\left[\sum_{i=1}^{d}\partial_i f(x)\, \Delta X^{(i)} \right] \\
		&\quad + \frac{1}{2}\mathbb{E}_x\left[\sum_{i,j=1}^{d}\partial_{ij} f(x)\, \Delta X^{(i)} \Delta X^{(j)}\right] \\
		&\quad + \frac{1}{6}\mathbb{E}_x\left[\sum_{i,j,k=1}^{d}\partial_{ijk} f(x)\, \Delta X^{(i)} \Delta X^{(j)} \Delta X^{(k)}\right] \\
		&\quad + \frac{1}{24}\mathbb{E}_x\left[\sum_{i,j,k,l=1}^{d}\partial_{ijkl} f(x)\, \Delta X^{(i)} \Delta X^{(j)} \Delta X^{(k)} \Delta X^{(l)}\right]\\
		&\quad + \frac{1}{120}\mathbb{E}_x\left[\sum_{i,j,k,l,m=1}^{d}\partial_{ijklm} f(x)\, \Delta X^{(i)} \Delta X^{(j)} \Delta X^{(k)} \Delta X^{(l)}\Delta X^{(m)}\right]\\
		&\quad + \frac{1}{5!}\mathbb{E}_x\left[
		\int_0^1 (1-s)^5
		\sum_{i,j,k,l,m,n=1}^d
		\partial_{ijklmn}f(x+s\Delta X)\,
		\Delta X^{(i)}\Delta X^{(j)}\Delta X^{(k)}\Delta X^{(l)}\Delta X^{(m)}\Delta X^{(n)}
		\,ds\right].
	\end{align*}
	Throughout the proof, $O(h^3)$ denotes a quantity bounded in absolute value by
	$h^3$ times a polynomial in $\|x\|$, uniformly in $h\in(0,1)$.
	By considering the linear term in the Taylor series expansion,  we write
	\begin{align*}
		\mathbb{E}_x\left[ \sum_{i=1}^{d}\partial_i f(x) \Delta X^{(i)}\right] 
		&= \sum_{i=1}^{d}\partial_i f(x) \left( h \,\drift^{(i)}(x)
		\;+\; \sqrt{2h}\,\diff^{(i)}(x)\,\mathbb{E} Z_1^{(i)}\right)\\
		&= h \sum_{i=1}^{d}\partial_i f(x)  \drift^{(i)}(x).
	\end{align*}
	With the help of \eqref{eq:A-R_expansion} and \eqref{eq:RA}, we obtain
	\begin{align}\label{eq:linearterm}
		\mathbb{E}_x\left[ \sum_{i=1}^{d}\partial_i f(x)\,\Delta X^{(i)} \right]
		&= h \sum_{i=1}^{d}\partial_i f(x)\bigl(\nabla_i\log\pi(x) + R^{(i)}_{A}(x,h)\bigr) \nonumber\\
		&= h\,\nabla f(x)\cdot\nabla\log\pi(x)
		\;+\; h \sum_{i=1}^{d}\partial_i f(x)\,R^{(i)}_{A}(x,h)\nonumber\\
		&= h\,\nabla f(x)\cdot\nabla\log\pi(x)  - \frac{c}{2}h^2
		\sum_{i=1}^d \partial_i f(x)\bigl(\nabla_i\log\pi(x)\bigr)^3
		+ O(h^3).
	\end{align}
	Now, from the quadratic term in the Taylor series expansion, we have
	\begin{align}\label{eq:qudraticterm}
		\frac{1}{2}\mathbb{E}_x\left[ \sum_{i,j=1}^{d}\partial_{ij} f(x) \Delta X^{(i)} \Delta X^{(j)}  \right]  
		&= \frac{1}{2}\, \text{tr}\left(\nabla^2 f(x) \mathbb{E}_x \left[\Delta X \Delta X^\top  \right]\right),
	\end{align}
	where tr\,$(\cdot)$ is the trace operator. Using $\mathbb{E}[Z_1] =0$ and $\mathbb{E}[Z_1 Z_1^\top] =I_d$, we get
	\begin{align*}
		\mathbb{E}_x \left[\Delta X \Delta X^\top  \right] 
		&=  \mathbb{E}_x \left[ \left(h\drift(x)
		\;+\; \sqrt{2h}\,\diff(x) Z_1\right)\left(h\drift(x)
		\;+\; \sqrt{2h}\,\diff(x) Z_1\right)^\top\right]\\
		&=  h^2 \drift(x)\drift(x)^\top + \sqrt{2h}\, h \, \drift(x) \mathbb{E}[Z_1^\top] \diff(x)^\top \\
		&\;\;\;\;+ \sqrt{2h}\, h \,\diff(x) \mathbb{E}[Z_1] \drift(x)^\top + 2 h\, \diff(x) \mathbb{E}[Z_1 Z_1^\top] \diff(x)^\top\\
		&= h^2 \drift(x)\drift(x)^\top + 2 h \,\diff(x)^2.
	\end{align*}
	By the
	properties of the trace operator, (\ref{eq:qudraticterm}) simplifies to
	\begin{align*}
		\frac{1}{2}\, \text{tr}\left(\nabla^2 f(x) \mathbb{E}_x \left[\Delta X \Delta X^\top  \right]\right) 
		&= \frac{1}{2}\, \text{tr}\left(\nabla^2 f(x) h^2 \drift(x)\drift(x)^\top \right) +  \frac{1}{2}\, \text{tr}\left(\nabla^2 f(x) 2 h \,\diff(x)^2 \right)\nonumber\\
		&=  \frac{h^2}{2}\, \text{tr}\left(\drift(x)^\top \nabla^2 f(x) \drift(x) \right) \;+\; h \sum_{i=1}^{d}  \partial_{ii}f(x) \diff^{(i)}(x)^2 \notag\\
		&= \frac{h^2}{2} \drift(x)^\top \,\nabla^2 f(x)\,\drift(x) 	+ h\sum_{i=1}^{d}  \partial_{ii}f(x) \left(\diff^{(i)}(x)^2 -1\right) + \nabla \cdot \nabla f(x) h \nonumber\\ 
		&= \frac{h^2}{2} \drift(x)^\top \,\nabla^2 f(x)\,\drift(x)- \frac{h^2}{2}\sum_{i=1}^{d}  \partial_{ii}f(x) \drift^{(i)}(x)^2 + \nabla \cdot \nabla f(x) h,
	\end{align*}
	where the last step due to (\ref{eq:SULA_drift}). Further, from \eqref{eq:A-R_expansion} and \eqref{eq:RA}, we get
	\begin{align}\label{eq:quadratic-term}
		\frac{1}{2}\mathbb{E}_x\left[ \sum_{i,j=1}^{d}\partial_{ij} f(x) \Delta X^{(i)} \Delta X^{(j)}  \right]  
		&= \nabla \cdot \nabla f(x)\, h + \frac{h^2}{2} \sum_{i,j=1}^{d}\partial_{ij}f(x)\, \nabla_i \log \pi(x)\nabla_j \log \pi(x)\nonumber\\
		&\;\;\;-\frac{h^2}{2} \sum_{i=1}^{d}\partial_{ii}f(x)(\nabla_i \log \pi(x))^2 + O(h^3).
	\end{align}
	Now consider the cubic term in the Taylor series
	\begin{align*}
		\frac{1}{6}\,\mathbb{E}_x\Bigg[\sum_{i,j,k=1}^d \partial_{ijk} f(x) &\Delta X^{(i)} \Delta X^{(j)} \Delta X^{(k)}  \Bigg]
		= \frac{1}{6} \sum_{i,j,k=1}^d \partial_{ijk} f(x)\, \mathbb{E}_x\Big[
		\left(h \drift^{(i)}(x) + \sqrt{2h} \,\diff^{(i)}(x) Z_1^{(i)}\right) \\
		&\quad \times \left(h \drift^{(j)}(x) + \sqrt{2h} \,\diff^{(j)}(x) Z_1^{(j)}\right)
		\left(h \drift^{(k)}(x) + \sqrt{2h} \,\diff^{(k)}(x) Z_1^{(k)}\right)
		\Big].
	\end{align*}
	The above summation involves $8$ terms. Since $Z_1$ is a zero mean Gaussian vector, therefore $\mathbb{E}_x\big[\diff(x)Z_1\big] = 0$
	and also
	$$
	\mathbb{E}_x\left[\left(\diff^{(i)}(x)Z^{(i)}_1\right)\left(\diff^{(j)}(x)Z^{(j)}_1\right)\left(\diff^{(k)}(x)Z^{(k)}_1\right)\right] = 0, \text{ for all } i,j,k\in\{1,\ldots,d\}.
	$$
	Consequently, any term containing an odd number of factors of 
	$(\diff(x)Z_1)$ vanishes. Thus, only the term with no such factor
	and exactly three terms with two Gaussian factors remain. Since
	$f\in C_P^6(\mathbb R^d)$, we have
	$\partial_{ijk} f = \partial_{jik} f = \partial_{ikj} f$,
	so these three terms are identical. Hence
	\vspace{.4cm}\\
	$\displaystyle\frac{1}{6}\,\mathbb{E}_x\Bigg[\sum_{i,j,k=1}^d \partial_{ijk} f(x) \Delta X^{(i)} \Delta X^{(j)} \Delta X^{(k)}  \Bigg]$
	\begin{align*}
		&= \frac{h^3}{6} \sum_{i,j,k=1}^d \partial_{ijk} f(x)\, 
		\drift^{(i)}(x) \drift^{(j)}(x) \drift^{(k)}(x)\\ 
		&\;\;+ 3 \cdot \frac{1}{6}\sum_{i,j,k=1}^d \partial_{ijk} f(x)\mathbb{E}_x\left[
		\left(h\,\drift^{(i)}(x)\right)
		\left(\sqrt{2h}\,\diff^{(j)}(x)Z^{(j)}_1\right)
		\left(\sqrt{2h}\,\diff^{(k)}(x)Z^{(k)}_1\right)
		\right] \\
		&= \frac{h^3}{6} \sum_{i,j,k=1}^d \partial_{ijk} f(x)\, 
		\drift^{(i)}(x) \drift^{(j)}(x) \drift^{(k)}(x)\\ 
		&\;\;+ h^2\sum_{i,j,k=1}^d \partial_{ijk} f(x)\, \drift^{(i)}(x)\diff^{(j)}(x)\diff^{(k)}(x)\delta_{jk},
	\end{align*}
	where $\delta_{jk}$ denotes the Kronecker delta function which takes value $1$ when $j=k$ and $0$ otherwise. Therefore, the cubic term we get
	\begin{align}\label{eq:cubicterm}
		= \frac{h^3}{6} \sum_{i,j,k=1}^d \partial_{ijk} f(x)\, 
		\drift^{(i)}(x)\drift^{(j)}(x)\drift^{(k)}(x) 
		+ h^2 \sum_{i,j=1}^d \partial_{ijj} f(x)\drift^{(i)}(x)\left(\diff^{(j)}(x)\right)^2.
	\end{align}
	
	Using \eqref{eq:SULA_drift} and \eqref{eq:A-R_expansion}, we obtain
	\begin{align}\label{eq:cubic-second}
		\frac{1}{6}\mathbb{E}_x\Bigg[\sum_{i,j,k=1}^d \partial_{ijk} f(x)
		\Delta X^{(i)} \Delta X^{(j)} \Delta X^{(k)}\Bigg]
		&=  h^2 \sum_{i,j=1}^d \partial_{ijj} f(x)\drift^{(i)}(x)\left(\diff^2(x)\right)_{jj} + O(h^3)\nonumber\\
		&= h^2
		\sum_{i,j=1}^d \nabla_i\log\pi(x)\,\partial_{ijj} f(x)
		+ O(h^3).
	\end{align}
	
	We now consider the fourth order term from the Taylor series. Products with
	one drift factor $h\,\drift(x)$ and three Gaussian factors have zero
	expectation, and all other terms containing at least one drift factor are of order
	$h^3$ or higher. Therefore we write\vspace{.4cm}\\
	$
	\frac{1}{24}\,
	\mathbb{E}_x\Bigg[
	\sum_{i,j,k,l=1}^d
	\partial_{ijkl} f(x)\,
	\Delta X^{(i)} \Delta X^{(j)} \Delta X^{(k)} \Delta X^{(l)}
	\Bigg]$
	\begin{align*}
		&=
		\frac{(2h)^2}{24}
		\sum_{i,j,k,l=1}^d
		\partial_{ijkl} f(x)\,
		\mathbb{E}\left[
		\left(\diff^{(i)}(x)Z^{(i)}_1\right)\left(\diff^{(j)}(x)Z^{(j)}_1\right)\left(\diff^{(k)}(x)Z^{(k)}_1\right)
		\left(\diff^{(l)}(x)Z^{(l)}_1\right)
		\right]
		+ O(h^3).
	\end{align*}
	Using Wick's formula, it simplifies to
	\begin{align*}
		&= \frac{(2h)^2}{24}
		\sum_{i,j,k,l=1}^d
		\partial_{ijkl} f(x)\,
		\diff^{(i)}(x)
		\diff^{(j)}(x)
		\diff^{(k)}(x)
		\diff^{(l)}(x)
		\mathbb{E}\!\left[Z_1^{(i)} Z_1^{(j)} Z_1^{(k)} Z_1^{(l)}\right]
		+ O(h^3)\\
		&= \frac{h^2}{6}
		\sum_{i,j,k,l=1}^d
		\partial_{ijkl} f(x)\,
		\diff^{(i)}(x)
		\diff^{(j)}(x)
		\diff^{(k)}(x)
		\diff^{(l)}(x)
		\Big(
		\delta_{ij}\delta_{kl}
		+
		\delta_{ik}\delta_{jl}
		+
		\delta_{il}\delta_{jk}
		\Big)
		+ O(h^3).
	\end{align*}
	
	Since $f\in C_P^{6}(\mathbb{R}^d)$, we have
	$
	\partial_{iikk} f
	=
	\partial_{ikik} f
	=
	\partial_{ikki} f,
	$
	so all three resulting sums are identical. Using the Kronecker delta structure,
	and $\diff^{(i)}(x)^2=1-\tfrac h2\,\drift^{(i)}(x)^2$ from
	\eqref{eq:SULA_drift}, so that
	$\diff^{(i)}(x)^2\diff^{(k)}(x)^2=1+O(h)$, we obtain
	\begin{align}
		\label{eq:fourth-second}
		\frac{1}{24}\mathbb{E}_x\!\Bigg[
		\sum_{i,j,k,l=1}^d
		\partial_{ijkl} f(x)\,
		\Delta X^{(i)} \Delta X^{(j)} \Delta X^{(k)} \Delta X^{(l)}
		\Bigg]
		=
		\frac{h^2}{2}
		\sum_{i,k=1}^d
		\partial_{i i k k} f(x)
		+ O(h^3).
	\end{align}
	
	It remains to bound the fifth-order term and the remainder. In the
	fifth-order term, a product containing an odd number of Gaussian
	factors has zero expectation. Every surviving product therefore
	contains at least one drift factor. The lowest power of $h$ occurs
	with one drift factor and four Gaussian factors, giving
	$h(\sqrt{h})^4=h^3$. Since $\|\drift(x)\|\le\|\nabla\log\pi(x)\|$
	and $\|\diff(x)Z_1\|\le\|Z_1\|$, the expected fifth-order term is
	bounded in absolute value by $h^3$ times a polynomial in $\|x\|$.
	For the integral remainder, polynomial growth of $\nabla^6 f$ gives
	an absolute bound of the form
	$
	C\bigl(1+\|x\|^q+\|\Delta X\|^q\bigr)\|\Delta X\|^6
	$
	for some $C,q>0$. For every $r\ge 1$ and $0<h<1$,
	\[
	\mathbb E_x\|\Delta X\|^r
	\le C_r h^{r/2}\bigl(1+\|\nabla\log\pi(x)\|^r\bigr).
	\]
	Applying this bound with $r=6$ and $r=6+q$ shows that the expected
	remainder is also bounded by $h^3$ times a polynomial in $\|x\|$.
	The coefficient expansions used in the first four terms have the
	same polynomial control: for $v\ge0$,
	\[
	\bigl|(1+v)^{-1/2}-1+v/2\bigr|\le 3v^2/8,
	\qquad v=ch|\nabla_i\log\pi(x)|^2.
	\]
	Combining the terms proves \eqref{eq:RoTULA-secondorder-expansion}
	with $|M_2(x,h)|\le h^3R_2(x)$, where $R_2$ may be chosen of the
	form $C(1+\|x\|^2)^m$ for some integer $m$.
	For part (i), take $f\in C_P^4(\mathbb R^d)$ and expand through
	degree three with a fourth-order integral remainder. The first two
	terms give $h\mathcal Lf(x)$ up to an error of order $h^2$ with
	polynomial growth in $x$. The expected cubic term has the same bound
	by Gaussian symmetry, and the fourth-order remainder satisfies it by
	the preceding moment estimate. This gives
	$|M_1(x,h)|\le h^2R_1(x)$ for some $R_1\in C_P^\infty(\mathbb R^d)$.
\end{proof}

The next lemma establishes the bound used in the proof of
Propositions~\ref{prop:mombound} and~\ref{prop:RoTULA-drift}.

\begin{lemma}\label{lem:coercivity}
Under Assumption~\ref{assum1}, there exist constants $M>0$ and
$\kappa'>0$ such that, for every $\|x\|\ge M$ and every $h\in(0,1)$,
\begin{equation}\label{eq:meanmatcheddrift}
	\left\langle \frac{x}{\|x\|},\,\drift(x)\right\rangle
	+ \frac{h}{2\|x\|}\|\drift(x)\|^2
	\le -\kappa'.
\end{equation}
\end{lemma}

\begin{proof}
Fix $h\in(0,1)$. For every $u\in\mathbb{R}$,
$$
\frac{u}{\sqrt{1+ch u^2}}
= \operatorname{sign}(u)
\int_0^{|u|}\frac{1}{(1+ch s^2)^{3/2}}\,ds
= \operatorname{sign}(u)
\int_0^{\infty}\frac{\mathbf{1}\{|u|>s\}}{(1+ch s^2)^{3/2}}\,ds.
$$
Applying this identity to $u=\nabla_i\log\pi(x)$, multiplying by
$x_i$, and summing over $i=1,\dots,d$ gives
$$
\langle x,\drift(x)\rangle
= \int_0^\infty \frac{1}{(1+ch s^2)^{3/2}}
\sum_{i:\,|\nabla_i\log\pi(x)|>s}
x_i\,\operatorname{sign}\bigl(\nabla_i\log\pi(x)\bigr)\,ds.
$$
For $\|x\|\ge R_0$, Assumption~\ref{assum1} bounds the sum by
$-\kappa_0\|x\|$ on $[0,s_0]$ and by zero on $(s_0,\infty)$.
Hence, for $\|x\|\ge R_0$,
$$
\left\langle \frac{x}{\|x\|},\,\drift(x)\right\rangle
\le -\kappa_0\int_0^{s_0}\frac{1}{(1+ch s^2)^{3/2}}\,ds
= -\frac{\kappa_0s_0}{\sqrt{1+ch s_0^2}}
\le -\kappa,
\qquad
\kappa := \frac{\kappa_0s_0}{\sqrt{1+c s_0^2}}>0,
$$
where the last inequality uses $h<1$.
Finally, the definition of $\drift$ gives
$h|\drift^{(i)}(x)|^2\le 1/c$ for every $i$, so
$h\|\drift(x)\|^2\le d/c$. Taking
$M=\max\{R_0,d/(c\kappa)\}$, we obtain, for every $\|x\|\ge M$,
$$
\left\langle \frac{x}{\|x\|},\,\drift(x)\right\rangle
+ \frac{h}{2\|x\|}\|\drift(x)\|^2
\le -\kappa+\frac{d}{2c\|x\|}
\le -\frac{\kappa}{2}.
$$
This proves \eqref{eq:meanmatcheddrift} with $\kappa'=\kappa/2$.
\end{proof}

The following proposition is important to prove the main weak convergence result.
\begin{proposition}\label{prop:mombound}
Under Assumption~\ref{assum1}, for any integer $p\ge1$ and finite horizon $T>0$ there exists constant $M'= M'(p,x, T)>0$ such that for every time-step $h\in(0,1)$ and every $t$ with $t\le T/h$,
$$
\mathbb{E}_x\big[ \|X_{t}\|^{2p} \big] \le M'.
$$
\end{proposition}
\begin{proof}
Recall that the one-step update of SULA with $X_0=x$ is
\[
X_1=x+h\,\drift(x)+\sqrt{2h}\,\diff(x)Z_1,
\qquad Z_1\sim\mathcal{N}(0,I_d).
\]
Write $w=x+h\,\drift(x)$. For $\|x\|\ge M$, multiplying
\eqref{eq:meanmatcheddrift} in Lemma~\ref{lem:coercivity} by
$2h\|x\|$ gives
\[
2h\langle x,\drift(x)\rangle+h^2\|\drift(x)\|^2
\le -2h\kappa'\|x\|.
\]
Hence $\|w\|^2\le\|x\|^2$ for $\|x\|\ge M$. For $\|x\|\le M$,
let $G_M=\sup_{\|y\|\le M}\|\nabla\log\pi(y)\|$, which is finite
by continuity. Since $\|\drift(x)\|\le\|\nabla\log\pi(x)\|$,
we have
\[
\|w\|^2
\le \|x\|^2+2hMG_M+h^2G_M^2
\le \|x\|^2+b_0h,
\qquad b_0=2MG_M+G_M^2.
\]
Combining the two cases, $\|w\|^2\le\|x\|^2+b_0h$ for every
$x\in\mathbb{R}^d$ and $h\in(0,1)$.
Now, for any integer $p\ge1$,
\[
\|X_1\|^{2p}
=
\left(
\|w\|^2
+2\sqrt{2h}\,\langle w,\diff(x)Z_1\rangle
+2h\|\diff(x)Z_1\|^2
\right)^p.
\]
Expanding the right-hand side, all terms of odd degree in $Z_1$
have zero expectation by symmetry. The term of degree zero is
$\|w\|^{2p}$. For $k=1,\dots,p$, each term of degree $2k$ in $Z_1$
is bounded in absolute value by a constant times
$
h^k\|w\|^{2p-2k}\|\diff(x)Z_1\|^{2k}.
$
Since $\diff(x)$ is diagonal with entries in $[0,1]$, we have
$\|\diff(x)Z_1\|\le\|Z_1\|$. Taking expectations and using the
finiteness of the Gaussian moments gives
\[
\mathbb{E}_x[\|X_1\|^{2p}]
\le \|w\|^{2p}
+C_p\sum_{k=1}^p h^k\|w\|^{2p-2k},
\]
where $C_p$ is independent of $x$ and $h$. Using $\|w\|^2\le\|x\|^2+b_0h$ and expanding by the binomial theorem,
the leading term is $\|x\|^{2p}$, and every remaining term is of the
form (a constant times) $h^{n}\|x\|^{2q}$ with $n\ge1$ and $0\le q\le p$.
Since $h<1$ gives $h^{n}\le h$, and $\|x\|^{2q}\le 1+\|x\|^{2p}$,
it follows that, for some $c'>0$ independent of $x$ and $h$,
\[
\mathbb{E}_x[\|X_1\|^{2p}]
\le (1+c'h)\|x\|^{2p}+c'h.
\]

Finally, by the Markov property, for every integer $t\ge1$,
\[
1+\mathbb{E}_x[\|X_t\|^{2p}]
\le (1+c'h)\bigl(1+\mathbb{E}_x[\|X_{t-1}\|^{2p}]\bigr).
\]
Iterating this inequality and using $th\le T$, we obtain
\[
1+\mathbb{E}_x[\|X_t\|^{2p}]
\le (1+c'h)^t(1+\|x\|^{2p})
\le e^{c'T}(1+\|x\|^{2p}).
\]
Taking $M'=e^{c'T}(1+\|x\|^{2p})$ proves the result.
\end{proof}

We now state the following two lemmas which are used in proving Proposition \ref{sgprop1} .

\begin{lemma}\label{lmm:propSG1}
Let $\Psi(x) = \dfrac{x}{\sqrt{1 + c x^2}}$ with $c > 0$ and $x \in \mathbb{R}$. Then for any $\delta \in \mathbb{R}$ and $\eta \neq 0$,
$$
\frac{1}{2}\bigl(\Psi(\delta+\eta) + \Psi(\delta-\eta)\bigr) 
\begin{cases}
	< \Psi(\delta), & \text{if } \delta > 0, \\
	= \Psi(\delta), & \text{if } \delta = 0, \\
	> \Psi(\delta), & \text{if } \delta < 0.
\end{cases}
$$ For $\eta=0$, the left-hand side equals $\Psi(\delta)$ for every $\delta\in\mathbb{R}$.
\end{lemma}

\begin{proof}
We first compute the derivative of $\Psi$
$$
\Psi'(x) = (1 + c x^2)^{-3/2}.
$$
Note that $\Psi'(x) > 0$ for all $x$, $\Psi'$ is even, and $\Psi'$ is strictly decreasing for $|x| \ge 0$. Since the expression $\Psi(\delta+\eta) + \Psi(\delta-\eta)$ is invariant under the symmetry of $\eta$, it is sufficient to consider $\eta > 0$. Using the fundamental theorem of calculus, we write
$$
\Psi(\delta + \eta) = \Psi(\delta) + \int_{\delta}^{\delta + \eta} \Psi'(x)\,dx,
$$
$$
\Psi(\delta - \eta) = \Psi(\delta) - \int_{\delta - \eta}^{\delta} \Psi'(x)\,dx.
$$
Averaging these two expressions gives
$$
\frac{1}{2}\bigl(\Psi(\delta+\eta) + \Psi(\delta-\eta)\bigr)
= \Psi(\delta) + \frac{1}{2}\left(\int_{\delta}^{\delta + \eta} \Psi'(x)\,dx - \int_{\delta - \eta}^{\delta} \Psi'(x)\,dx\right).
$$
By a change of variables, we obtain
\begin{align}\label{eq:avg-psi}
	\frac{1}{2}\bigl(\Psi(\delta+\eta) + \Psi(\delta-\eta)\bigr)
	&= \Psi(\delta) + \frac{1}{2}\left(\int_0^\eta \Psi'(\delta + s)\,ds - \int_0^\eta \Psi'(\delta - s)\,ds\right) \nonumber\\
	\frac{1}{2}\bigl(\Psi(\delta+\eta) + \Psi(\delta-\eta)\bigr)
	&= \Psi(\delta) + \frac{1}{2}\int_0^\eta \bigl(\Psi'(\delta + s) - \Psi'(\delta - s)\bigr)\,ds.
\end{align}
For $s \in (0,\eta)$, since $\Psi'$ depends only on $|x|$ and is strictly decreasing in $|x|$. Observe that
$$
|\delta - s|^2 - |\delta + s|^2
= (\delta - s)^2 - (\delta + s)^2
= -4\delta s.
$$
Since $s > 0$, the sign of $-4\delta s$ is the opposite of the sign of $\delta$. Hence, If $\delta > 0$, then $|\delta - s| < |\delta + s|$, so $\Psi'(\delta - s) > \Psi'(\delta + s)$. The other case follows analogously.  Hence, from \eqref{eq:avg-psi}, the required result follows.
\end{proof}

\begin{lemma}\label{lmm:propSG2}
Let $\Psi(x) = \dfrac{x}{\sqrt{1 + c x^2}}$ with $c > 0$ and $x \in \mathbb{R}$. Then for any $\delta,\eta\in\mathbb{R}$
$$
\frac{1}{2}\bigl(\Psi(\delta+\eta) + \Psi(\delta-\eta)\bigr)
\begin{cases}
	> 0, & \text{if } \delta > 0, \\
	= 0, & \text{if } \delta = 0, \\
	< 0, & \text{if } \delta < 0.
\end{cases}
$$
\end{lemma}

\begin{proof} Note that $\Psi$ is odd and strictly increasing on $\mathbb{R}$, since
$$
\Psi'(x)=(1+cx^2)^{-3/2}>0.
$$
For $\eta=0$, the expression equals $\Psi(\delta)$ and therefore has the same sign as $\delta$.
Also the expression $\Psi(\delta+\eta) + \Psi(\delta-\eta)$ is invariant under symmetry of $\eta$, it is sufficient to consider $\eta > 0$ case. For $\delta = 0$, the proof holds trivially since $\Psi$ is an odd function. Let $\delta > 0$. For $\delta \ge \eta$, $\delta-\eta \ge 0$ and $\delta+\eta > 0$. Since $\Psi$ is increasing and $\Psi(0)=0$, we have
$$
\Psi(\delta+\eta) > \Psi(\delta-\eta) \ge 0.
$$ 
Also, when $0 < \delta < \eta$, $\delta-\eta < 0 < \delta+\eta$. Since $\Psi$ is odd, we have
$$
\Psi(\delta-\eta) = -\Psi(\eta-\delta),
$$
so
$$
\Psi(\delta+\eta) + \Psi(\delta-\eta)
= \Psi(\delta+\eta) - \Psi(\eta-\delta).
$$
Since $\Psi$ is strictly increasing, this is positive if and only if
$$
\delta+\eta > \eta-\delta \quad \Longleftrightarrow \quad 2\delta > 0,
$$
which holds since $\delta > 0$. For $\delta < 0$. Let $\delta' = -\delta > 0$. Since $\Psi$ is odd function. Hence
$$
\Psi(\delta+\eta) + \Psi(\delta-\eta)
= -\left(\Psi(\delta'+\eta) + \Psi(\delta'-\eta)\right).
$$
By previous case $\delta' > 0$, the quantity $\Psi(\delta'+\eta) + \Psi(\delta'-\eta)$ is strictly positive. Thus, the lemma holds in all cases, with equality if $\delta = 0$.
\end{proof}

For a symmetric positive definite matrix $\Sigma\in\mathbb{R}^{d\times d}$, define
\[
\delta(\Sigma) := \min_{j}\Bigl(\Sigma_{jj}-\sum_{i\neq j}|\Sigma_{ij}|\Bigr),
\qquad
\rho(\Sigma) := \max_i\sum_{j}|\Sigma_{ij}|,
\qquad
K_\Sigma := \sum_{i,j}|\Sigma_{ij}|.
\]
We say that $\Sigma$ is strictly diagonally dominant if $\delta(\Sigma)>0$.

\begin{proposition}[Sufficient conditions for Assumption~1(a)]\label{prop:sufficient}
Assumption~\ref{assum1} holds in each of the following cases.
\begin{enumerate}
	\item[(a)] \emph{Coordinatewise sign condition.} There exist $R_1,\eta>0$ such
	that, for all $\|x\|\ge R_1$, $x_i\nabla_i\log\pi(x)\le 0$ for every $i$ and
	$\langle x,\nabla\log\pi(x)\rangle\le-\eta\|x\|\,\|\nabla\log\pi(x)\|$; moreover
	$\|\nabla\log\pi(x)\|\to\infty$ as $\|x\|\to\infty$.
	\item[(b)] \emph{Product targets centred at the origin.}
	$\log\pi(x)=\sum_{i=1}^d\ell_i(x_i)$, where each $g_i:=\ell_i'$ is continuous,
	$u\,g_i(u)<0$ for $u\neq 0$, and $\liminf_{|u|\to\infty}|g_i(u)|>0$.
	\item[(c)] \emph{Product targets with superlinear tails.}
	$\log\pi(x)=\sum_{i=1}^d\ell_i(x_i)$, where each $g_i:=\ell_i'$ is continuous,
	and for some $r>0$, $u\,g_i(u)<0$ for $|u|\ge r$ and $|g_i(u)|\to\infty$ as
	$|u|\to\infty$.
	\item[(d)] \emph{Gaussian targets with bounded perturbation.}
	$\nabla\log\pi(x)=-Px-e(x)$, where $P\succ 0$, $\Sigma:=P^{-1}$ is strictly
	diagonally dominant, and $\|e(x)\|_\infty\le E<\infty$ for all $\|x\|\ge R_e$.
\end{enumerate}
\end{proposition}

\begin{proof}
Write $g:=\nabla\log\pi$ and, for $x\in\mathbb{R}^d$ and $s\ge 0$,
\[
T_s(x):=\{i:\ |g_i(x)|>s\},
\qquad
\Gamma_s(x):=\sum_{i\in T_s(x)} x_i\operatorname{sign}\bigl(g_i(x)\bigr),
\]
so that Assumption~\ref{assum1} requires $\Gamma_s(x)\le-\kappa_0\|x\|$ for
$s\in[0,s_0]$ and $\Gamma_s(x)\le 0$ for $s>s_0$, whenever $\|x\|\ge R_0$.
Note that if $x_ig_i(x)\le 0$ and $g_i(x)\neq 0$, then
$x_i\operatorname{sign}(g_i(x))=-|x_i|$.

\medskip
\noindent\emph{Case (a).} Fix $\|x\|\ge R_1$. By the observation above, every
term of $\Gamma_s(x)$ equals $-|x_i|\le 0$, so $\Gamma_s(x)\le 0$ for all $s\ge 0$.
Let $s_0:=1$ and $s\in[0,1]$. The coordinates outside $T_s(x)$ satisfy
\[
\sum_{i\notin T_s(x)}|x_i||g_i(x)|\le s\|x\|_1\le\sqrt{d}\,\|x\|.
\]
Since $x_ig_i(x)\le 0$ for all $i$, we have
$-\langle x,g(x)\rangle=\sum_i|x_i||g_i(x)|$, and hence
\[
\sum_{i\in T_s(x)}|x_i||g_i(x)|
= -\langle x,g(x)\rangle-\sum_{i\notin T_s(x)}|x_i||g_i(x)|
\ge \eta\|x\|\|g(x)\|-\sqrt{d}\,\|x\|.
\]
Using $|g_i(x)|\le\|g(x)\|$,
\[
-\Gamma_s(x)=\sum_{i\in T_s(x)}|x_i|
\ge\frac{1}{\|g(x)\|}\sum_{i\in T_s(x)}|x_i||g_i(x)|
\ge\|x\|\Bigl(\eta-\frac{\sqrt d}{\|g(x)\|}\Bigr)
\ge\frac{\eta}{2}\|x\|
\]
whenever $\|g(x)\|\ge 2\sqrt d/\eta$, which holds for all sufficiently large
$\|x\|$. Hence Assumption~\ref{assum1} holds with $\kappa_0=\eta/2$, $s_0=1$ and
$R_0$ large enough.

\medskip
\noindent\emph{Case (b).} Here $g_i(x)=g_i(x_i)$ and
$\operatorname{sign}(g_i(x_i))=-\operatorname{sign}(x_i)$, so every term of
$\Gamma_s(x)$ equals $-|x_i|\le 0$, and $\Gamma_s(x)\le 0$ for all $x$ and $s$.
Choose $s_0>0$ smaller than $\min_i\liminf_{|u|\to\infty}|g_i(u)|$. Then there is
$r>0$ such that $|g_i(u)|>s_0$ whenever $|u|>r$, for every $i$. For
$s\in[0,s_0]$, every coordinate with $|x_i|>r$ belongs to $T_s(x)$, and all terms
are non-positive, so
\[
-\Gamma_s(x)\ge\sum_{i:\,|x_i|>r}|x_i|\ge\|x\|_1-dr\ge\|x\|-dr\ge\tfrac12\|x\|
\qquad\text{for }\|x\|\ge 2dr.
\]
Hence Assumption~\ref{assum1} holds with $\kappa_0=1/2$ and $R_0=2dr$.

\medskip
\noindent\emph{Case (c).} Let
\[
K:=\max_i\sup_{|u|\le r}|g_i(u)|<\infty,
\qquad
\gamma:=\min_i\inf_{|u|\ge r}|g_i(u)|>0 .
\]
$K$ is finite by continuity. $\gamma$ is positive because $g_i$ is continuous and
does not vanish on $\{|u|\ge r\}$, and $|g_i(u)|\to\infty$, so the infimum is
attained. Choose $r'\ge r$ such that $|g_i(u)|>K$ for all $|u|\ge r'$ and all $i$.
Call coordinate $i$ \emph{good} if $|x_i|\ge r$ and \emph{bad} if $|x_i|<r$. A good
coordinate in $T_s(x)$ contributes $-|x_i|$ to $\Gamma_s(x)$; a bad coordinate
contributes at most $|x_i|<r$ and satisfies $|g_i(x_i)|\le K$. Let
$\|x\|\ge R_0:=\max\{\sqrt d\,r',\,d^{3/2}r,\,4dr\}$ and pick $j$ with
$|x_j|=\|x\|_\infty\ge\|x\|/\sqrt d\ge r'$. Then $j$ is good and $|g_j(x_j)|>K$.

Fix $s\ge 0$. If $T_s(x)$ contains no bad coordinate, all its terms are
non-positive and $\Gamma_s(x)\le 0$. If $T_s(x)$ contains a bad coordinate $i$,
then $s<|g_i(x_i)|\le K<|g_j(x_j)|$, so $j\in T_s(x)$ as well, and
\[
\Gamma_s(x)\le-|x_j|+\sum_{i\text{ bad}}|x_i|\le-\frac{\|x\|}{\sqrt d}+dr\le 0,
\]
because $\|x\|\ge d^{3/2}r$. Hence $\Gamma_s(x)\le 0$ for all $s\ge 0$.

Now let $s_0:=\gamma/2$ and $s\in[0,s_0]$. Every good coordinate has
$|g_i(x_i)|\ge\gamma>s$, so it belongs to $T_s(x)$. Hence
\[
\Gamma_s(x)\le-\sum_{i\text{ good}}|x_i|+\sum_{i\text{ bad}}|x_i|
=-\|x\|_1+2\sum_{i\text{ bad}}|x_i|
\le-\|x\|+2dr\le-\tfrac12\|x\|,
\]
because $\|x\|\ge 4dr$. Hence Assumption~\ref{assum1} holds with $\kappa_0=1/2$.

\medskip
\noindent\emph{Case (d).} Write $\delta,\rho$ for $\delta(\Sigma),\rho(\Sigma)$ and
let $\|x\|\ge R_e$. Put $y:=-g(x)=Px+e(x)$ and $\sigma_i:=\operatorname{sign}(y_i)$.
Then $\operatorname{sign}(g_i(x))=-\sigma_i$, $x=\Sigma(y-e(x))$, and, with
$T:=T_s(x)=\{i:|y_i|>s\}$,
\[
\Gamma_s(x)=-\sum_{i\in T}\sigma_ix_i .
\]
Since $|(\Sigma e(x))_i|\le\sum_j|\Sigma_{ij}||e_j(x)|\le\rho E$,
\begin{equation}\label{eq:suff-d1}
	\sum_{i\in T}\sigma_ix_i
	=\sum_{i\in T}\sigma_i(\Sigma y)_i-\sum_{i\in T}\sigma_i(\Sigma e(x))_i
	\ge\sum_{i\in T}\sigma_i(\Sigma y)_i-|T|\rho E .
\end{equation}
Expanding $(\Sigma y)_i=\Sigma_{ii}y_i+\sum_{j\neq i}\Sigma_{ij}y_j$, using
$\sigma_iy_i=|y_i|$ and $\sigma_i\Sigma_{ij}y_j\ge-|\Sigma_{ij}||y_j|$, splitting $j$
according to $j\in T$ or $j\notin T$, and exchanging $i$ and $j$ in the double sum
over $T\times T$ (using $|\Sigma_{ij}|=|\Sigma_{ji}|$), we obtain
\begin{equation}\label{eq:suff-d2}
	\sum_{i\in T}\sigma_i(\Sigma y)_i
	\ge\sum_{j\in T}|y_j|\Bigl(\Sigma_{jj}-\sum_{i\in T,\,i\neq j}|\Sigma_{ij}|\Bigr)
	-\sum_{i\in T}\sum_{j\notin T}|\Sigma_{ij}||y_j| .
\end{equation}
Each bracket is at least $\Sigma_{jj}-\sum_{i\neq j}|\Sigma_{ij}|\ge\delta>0$, and
$|y_j|\le s$ for $j\notin T$.

\emph{Large thresholds, $s\ge\rho E/\delta$.} If $T=\emptyset$, then
$\Gamma_s(x)=0$. Otherwise let $m:=\min_{i\in T}|y_i|>s$. Since the brackets are
positive and $|y_j|\le s<m$ for $j\notin T$, \eqref{eq:suff-d2} gives
\[
\sum_{i\in T}\sigma_i(\Sigma y)_i
\ge m\sum_{j\in T}\Bigl(\Sigma_{jj}-\sum_{i\in T,\,i\neq j}|\Sigma_{ij}|\Bigr)
-m\sum_{j\in T}\sum_{i\notin T}|\Sigma_{ij}|
= m\sum_{j\in T}\Bigl(\Sigma_{jj}-\sum_{i\neq j}|\Sigma_{ij}|\Bigr)
\ge m|T|\delta ,
\]
where the double sum over $T\times T^c$ was relabelled using symmetry. With
\eqref{eq:suff-d1},
\[
-\Gamma_s(x)\ge|T|(m\delta-\rho E)>|T|(s\delta-\rho E)\ge 0 .
\]

\emph{Small thresholds, $s\in[0,s_0]$ with $s_0:=\max\{1,\rho E/\delta\}$.} In
\eqref{eq:suff-d2}, the first sum is at least
$\delta\sum_{j\in T}|y_j|\ge\delta(\|y\|_1-ds_0)$, and the second is at most
$s_0K_\Sigma$. Moreover,
\[
\|x\|=\|\Sigma(y-e(x))\|\le\|\Sigma\|_2\bigl(\|y\|_1+\sqrt d\,E\bigr),
\qquad\text{so}\qquad
\|y\|_1\ge\frac{\|x\|}{\|\Sigma\|_2}-\sqrt d\,E .
\]
Combining with \eqref{eq:suff-d1} and $|T|\le d$,
\[
-\Gamma_s(x)\ge\frac{\delta}{\|\Sigma\|_2}\|x\|-C_0,
\qquad
C_0:=\delta\sqrt d\,E+\delta ds_0+s_0K_\Sigma+d\rho E .
\]
Hence $\Gamma_s(x)\le-\frac{\delta}{2\|\Sigma\|_2}\|x\|$ whenever
$\|x\|\ge 2\|\Sigma\|_2C_0/\delta$.

The two regimes together give $\Gamma_s(x)\le 0$ for every $s\ge 0$ (for
$s\le s_0$ this follows from the second regime, and for $s>s_0\ge\rho E/\delta$
from the first), and $\Gamma_s(x)\le-\kappa_0\|x\|$ for $s\in[0,s_0]$, with
\[
\kappa_0=\frac{\delta}{2\|\Sigma\|_2},
\qquad
s_0=\max\Bigl\{1,\frac{\rho E}{\delta}\Bigr\},
\qquad
R_0=\max\Bigl\{R_e,\frac{2\|\Sigma\|_2C_0}{\delta}\Bigr\}. \qedhere
\]
\end{proof}

\begin{remark}
Case~(d) with $e\equiv -P\mu$ covers every Gaussian target $N(\mu,\Sigma)$ with
$\Sigma$ strictly diagonally dominant, for any mean $\mu$. It also covers Bayesian
logistic regression with prior $N(0,\sigma_0^2I_d)$ and an arbitrary design
$\{a_n\}_{n=1}^N$: there $P=\sigma_0^{-2}I_d$ and
$e(x)=-\sum_{n=1}^N\bigl(z_n-(1+e^{-a_n^\top x})^{-1}\bigr)a_n$, with
$\|e(x)\|_\infty\le\max_i\sum_n|a_{ni}|$.
\end{remark}




\end{document}